\RequirePackage{fix-cm}
\documentclass[a4paper,10pt]{article}
\usepackage[a4paper, total={16cm, 23.6cm}, tmargin=2.5cm, headheight=26pt, twoside, lmargin=2.5cm]{geometry}

\usepackage[affil-it]{authblk}

\usepackage[american]{babel}

\usepackage{stmaryrd}
\usepackage{amsmath,bm}
\usepackage{epsfig,color}
\usepackage{amsthm}
\usepackage{amsfonts}
\usepackage{amssymb}
\usepackage{url}

\usepackage{adjustbox}

\usepackage{subcaption}
\usepackage{pgfplots}
\pgfplotsset{compat=1.13}
\usepackage{tikz-3dplot}
\usetikzlibrary{patterns}
\usepackage{tikz}
\usetikzlibrary[positioning]
\usetikzlibrary{intersections}
\usetikzlibrary{backgrounds}
\usetikzlibrary{decorations.pathreplacing,calligraphy}

\usepackage[utf8]{inputenc}
\usepackage{booktabs}
\usepackage{color}
\usepackage{graphicx}
\usepackage[procent]{overpic}
\usepackage{upgreek}
\usepackage{changes}
\usepackage{caption}
\usepackage{listings}
\usepackage{cite}
\usepackage{ifthen}
\newcommand{\moreDetailed}[3]{
  \ifthenelse{\equal{#1}{true}}{#2}{#3}
}

\newcommand{\R}{\ensuremath{\mathbb{R}}}
\newcommand{\N}{\ensuremath{\mathbb{N}}}

\newcommand{\oN}{{\overline{N}}}

\newcommand{\hv}{h_v}
\newcommand{\Ac}{A_c}
\newcommand{\cpf}{c_{p,f}}
\newcommand{\mc}{\dot{m}_c}

\def\vect#1{\mbox{\boldmath$ #1$}}                   

\newcommand{\bA}{{\vect A}}
\newcommand{\bb}{{\vect b}}
\newcommand{\bC}{{\vect C}}
\newcommand{\bc}{{\vect c}}
\newcommand{\bD}{{\vect D}}
\newcommand{\bI}{{\vect I}}
\newcommand{\bW}{{\vect W}}

\newcommand{\bs}{{\vect s}}

\newcommand{\bv}{{\vect v}}
\newcommand{\bV}{{\vect V}}

\newcommand{\bz}{{\vect z}}

\newcommand{\oa}{{\overline{a}}}
\newcommand{\ob}{{\overline{b}}}
\newcommand{\oC}{{\overline{C}}}
\newcommand{\oc}{{\overline{c}}}
\newcommand{\od}{{\overline{d}}}
\newcommand{\oD}{{\overline{D}}}
\newcommand{\of}{{\overline{f}}}

\newcommand{\oNN}{{\overline{N}}}
\newcommand{\orr}{{\overline{r}}}
\newcommand{\orho}{{\overline{\rho}}}
\newcommand{\oT}{{\overline{T}}}

\newcommand{\oy}{{\overline{y}}}
\newcommand{\oZ}{{\overline{Z}}}

\newcommand{\ooa}{{\oa}}
\newcommand{\oob}{{\ob}}
\newcommand{\oor}{{\orr}}

\newcommand{\ooC}{{\overline{\oC}}}
\newcommand{\ooc}{{\overline{\oc}}}
\newcommand{\ood}{{\overline{\od}}}

\newcommand{\ooN}{{\overline{\oNN}}}
\newcommand{\oorho}{{\overline{\orho}}}
\newcommand{\ooT}{{\overline{\oT}}}

\newcommand{\oooT}{{\overline{\ooT}}}

\newcommand{\cc}{{\check{c}}}

\newcommand{\diag}{\mbox{diag}}
\newcommand{\sign}{\mbox{sign}}

\newcommand{\eref}[1]{(\ref{#1})}

\newtheorem{remark}{Remark}
\newtheorem{proposition}{Proposition}
\newtheorem{lemma}{Lemma}

\begin{document}

\title{Analytical investigation of 1D Darcy-Forchheimer flow under local thermal nonequilibrium}


\author{S.~M\"uller}
\author{M.~Rom}



\affil{Institut f\"ur Geometrie und Praktische Mathematik, RWTH Aachen University, Templergraben~55, 52062 Aachen, Germany \authorcr rom@igpm.rwth-aachen.de, mueller@igpm.rwth-aachen.de}

\date{}

\maketitle

\begin{abstract}
In the context of transpiration cooling, a 1D porous medium model consisting of a temperature system   and a mass-momentum system is derived from the 2D/3D Darcy-Forchheimer equations. The temperatures of the coolant and the solid are assumed to be in local nonequilibrium. This system is analytically verified to have a unique solution. Transpiration cooling simulations are performed by a two-domain approach, coupling the assembled 1D porous medium solutions with 2D solutions of a hot gas flow solver. A comparison of results obtained by applying the 1D porous medium model with the temperature system either including or neglecting fluid heat conduction justifies the use of the latter simplified system.
\end{abstract}

\vspace*{0.5cm}


\section{Introduction}
\label{sec:Intro}



Transpiration cooling is considered a promising heat protection technique for high-temperature and high-velocity gas flows. Compared with other active methods such as film cooling, it provides a superior cooling efficiency~\cite{Eckert54,Laganelli70}. A generic configuration is sketched in Fig.~\ref{fig:2D_setting}, where a cooling gas, initially stored in a reservoir, is driven through a porous medium such that it is injected into the hot gas flow. This is done by establishing a pressure difference between the hot gas and the coolant side.
\begin{figure}[t]
  \centering
  \begin{tikzpicture}[scale=1.359]
    \draw[fill=gray!90] (0,2) -- (7,2) -- (7,2.5) -- (0,2.5) -- (0,2);
    \draw (0,0.5) -- (0,2);
    \draw (7,0.5) -- (7,2);

    \draw[dashed,fill=gray!30] (2,0) -- (5,0) -- (5,0.5) -- (2,0.5) -- (2,0);

    \draw[fill=gray!90] (0,0) -- (2,0) -- (2,0.5) -- (0,0.5) -- (0,0);
    \draw[fill=gray!90] (5,0) -- (7,0) -- (7,0.5) -- (5,0.5) -- (5,0);

    \draw[dashed] (2,0) -- (5,0);

    \draw[->] (0.6,1.5) -- (1,1.5);
    \draw[->] (0.6,1.25) -- (1,1.25);
    \draw[->] (0.6,1) -- (1,1);
    \node[above] at (0.8,1.53) {\small hot gas};




    \draw[->] (3.25,0.5) -- (3.25,0.9);
    \draw[->] (3.5,0.5) -- (3.5,0.9);
    \draw[->] (3.75,0.5) -- (3.75,0.9);

    \draw[->] (3.25,-0.4) -- (3.25,0);
    \draw[->] (3.5,-0.4) -- (3.5,0);
    \draw[->] (3.75,-0.4) -- (3.75,0);

    \node at (3.5,0.35) {\small porous};
    \node at (3.5,0.15) {\small medium};


    \node at (3.5,-0.6) {\small cooling gas};


    \node[below] at (4.5,0) {\small $\Gamma_{\text{R}}$};
    \node[above] at (4.5,0.48) {\small $\Gamma_{\text{Int}}$};
    
    \draw[->] (7.4,0) -- (7.7,0);
    \draw[->] (7.4,0) -- (7.4,0.3);
    \node[right] at (7.7,0) {\small $x$};
    \node[above] at (7.4,0.3) {\small $y$};
    
  \end{tikzpicture}
  \caption{Illustration of 2D transpiration cooling setup ($\Gamma_{\text{R}}$: reservoir boundary, $\Gamma_{\text{Int}}$: interface between hot gas and porous medium).}
  \label{fig:2D_setting}
\end{figure}
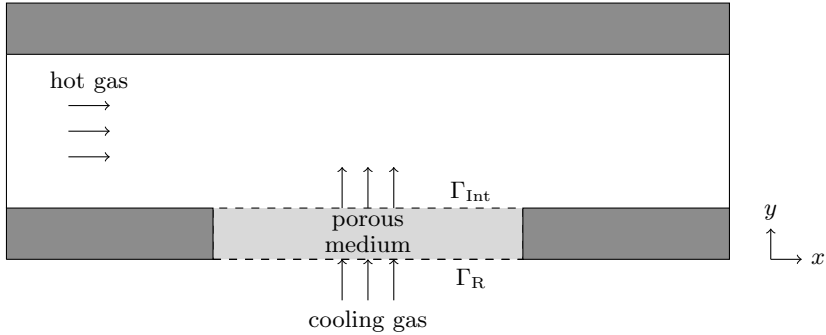


In our former works~\cite{Dahmen14_ijnmf,Dahmen15_ijhmt,Koenig19_jtht,Koenig20_book}, we present a two-domain approach for the numerical simulation of transpiration cooling problems, where a turbulent hot gas flow, e.g., a channel or a nozzle flow, is coupled with a laminar porous medium flow. For this purpose, we developed an iterative coupled procedure solving alternately the Reynolds-averaged Navier-Stokes equations in the hot gas domain and a system for mass, momentum and energy conservation in the porous medium domain. Momentum conservation is modeled by the Darcy-Forchheimer equation, and energy conservation is described by two temperature equations to account for a potential local thermal nonequilibrium between fluid and solid phase. 
For this coupled procedure, we developed appropriate coupling conditions, imposed at $\Gamma_{\text{Int}}$, see Fig.~\ref{fig:2D_setting}.
%

Using this two-domain approach, a wide variety of 2D and 3D test cases was investigated for different ranges of applications, cf.~\cite{Dahmen14_ijnmf,Dahmen15_ijhmt,Koenig19_jtht,Koenig20_book,Dahmen14_eccomas,Koenig18_aiaa}.
Although we found good agreement with experimental data, there are some numerical issues with respect to solving the porous medium flow. In particular, these concern the numerical stability of the solution method, especially in 3D, and the correct application of boundary conditions at the coolant reservoir and at the interface. Both problems contribute to an increase of the computation time. 
To overcome these issues and to generally accelerate computations, we recently developed in \cite{RomMueller22_ijhmt} a new coupled 2D/3D hot gas -- assembled-1D porous medium flow solver, where the 2D or 3D hot gas flow is coupled to a porous medium which is decomposed into several 1D problems.
These 1D problems are solved independently of each other and afterwards assembled to obtain a 2D or 3D solution for the porous medium flow.
We refer to this coupled model as \emph{assembled-1D model} in the following.
In 1D, the \emph{temperature system} (equations for fluid and solid temperature) can be decoupled from the \emph{mass-momentum system} (continuity and Darcy-Forchheimer momentum equation) and solved analytically. The mass-momentum system can be rewritten making use of conservation of mass to obtain an ordinary differential equation for the density (or the velocity). From the numerical solution of the latter, the velocity (or the density) can be directly evaluated.

In \cite{RomMueller22_ijhmt}, a 1D porous medium model is derived from a 2D/3D porous medium model.
In contrast to
previous work in the literature, where analytical solutions 
of the energy transport in 1D 
were in particular used to investigate the effect of different boundary conditions~\cite{Kubota76,Wolfersdorf05,Wang06,Wang08} or to predict wall temperatures in combination with a near-wall layer analysis of the hot gas flow~\cite{ZhangYi20,BaiYi21,NakayamaYi22},
we impose other  boundary conditions for the coupling with the hot gas flow.
Furthermore, the mass-momentum system is included in the solvability analysis in~\cite{RomMueller22_ijhmt} and in the current work.


For the assembled-1D model in \cite{RomMueller22_ijhmt}, we use a simplified 1D porous medium temperature system. This is derived from the 2D/3D porous medium model used in our works~\cite{Dahmen14_ijnmf,Dahmen15_ijhmt,Koenig19_jtht,Koenig20_book} or works by other authors~\cite{Munk19,Prokein20} by neglecting effects due to fluid heat conduction.
The simplification is based on the assumption
that the thermal conductivity of the fluid in the porous medium is much smaller than the thermal conductivity of the solid phase. The main objective of the current work is the justification of this simplification. For this purpose, we first recall the 1D model in Sect.~\ref{sec:PM-models} composed of the temperature system with fluid heat conduction and the mass-momentum system. In Sect.~\ref{sec:PM-full-temp-model-III}, we determine the exact analytical solution of the temperature system. This enters the mass-momentum system that is verified to have a unique solution, see Sect.~\ref{sec:PM-mass-momentum}. By means of coupled simulations with the assembled-1D model 
using either the 1D temperature system with fluid heat conduction or the simplified one without, we verify in Sect.~\ref{sec:numerical_results} that using the latter as in \cite{RomMueller22_ijhmt} will provide  accurate results for the range of parameters in our transpiration cooling  configuration.


\section{One-dimensional model for porous medium flow in the context of transpiration cooling}
\label{sec:PM-models}
In the following, we present the governing equations and boundary conditions for 1D stationary porous medium flow in local thermal nonequilibrium. The particular boundary conditions at the outlet of the porous medium establish the coupling with a hot gas flow for the simulation of transpiration cooling problems. More details on our original coupled two-domain approach can be found in~\cite{Dahmen14_ijnmf,Dahmen15_ijhmt,Koenig19_jtht,Koenig20_book}. The derivation of the 1D porous medium model and of the whole coupled assembled-1D model for transpiration cooling are presented in~\cite{RomMueller22_ijhmt}.




\textbf{Governing equations}.
With $\rho_f$, $v$, $\dot m_c$ and $A_c$ denoting the fluid density, Darcy velocity, prescribed coolant mass flow rate and cross-sectional area of the porous probe, respectively, the one-dimensional model is obtained taking advantage of 
%
\begin{equation}
\label{eq:mass-cons-III}
   \rho_f(y) \,v(y) =  \frac{\mc}{\Ac} 
\end{equation}
being constant due to conservation of mass. Then, the 1D system of equations for $y \in (0,L)$, where $L$ is the length of the porous probe, 
splits into a system for mass and momentum 
\begin{subequations}
\label{eq:mass-momentum-system}
\begin{align}
\rho_{f}^{\prime}(y) \, v(y) + \rho_{f}(y) \, v^{\prime}(y) &= 0, \label{eq:PM_conti_1D}\\
\varphi^{-2} \, \rho_{f}(y) \, v(y) \, v^{\prime}(y) &= -R \left( \rho_{f}^{\prime}(y) \, T_{f}(y) + \rho_{f}(y) \, T_{f}^{\prime}(y) \right) - \frac{\mu_f}{K_{D}} \, v(y) - \frac{\rho_{f}(y)}{K_{F}} \, v(y)^2, \label{eq:PM_Darcy_Forch_1D}
\end{align}
\end{subequations}
and a temperature system 
\begin{subequations}
\label{eq:temperature-system-full}
\begin{align}
-\varphi\, \kappa_f\, T_{f}^{\prime \prime}(y) + c_{p,f} \frac{\dot m_c}{A_c} T_{f}^{\prime}(y) &= h_v \left( T_s(y) - T_f(y) \right), 
\label{eq:PM_T_f_1D-full}\\
(1-\varphi) \, \kappa_s \, T_{s}^{\prime \prime}(y) &= h_v \left( T_s(y) - T_f(y) \right). 
\label{eq:PM_T_s_1D-full}
\end{align}
\end{subequations}
The unknowns are fluid density~$\rho_f$, Darcy velocity~$v$, fluid temperature~$T_f$ and solid temperature~$T_s$. All other quantities are given constant parameters, namely porosity~$\varphi$ of the porous probe, specific gas constant~$R$, dynamic viscosity~$\mu_f$ of the fluid, permeability~$K_D$, Forchheimer coefficient~$K_F$, specific heat capacity~$c_{p,f}$ of the fluid, coolant mass flow rate~$\dot m_c$, cross-sectional area~$A_c$ of the porous probe, volumetric heat transfer coefficient~$h_v$, thermal conductivity~$\kappa_f$ of the fluid and thermal conductivity~$\kappa_s$ of the solid. Note that $K_D$, $K_F$ and $\kappa_s$ are scalars in 1D and matrices in 2D/3D. Consequently, only the respective values in flow direction~$y$ of those matrices are used in 1D. In~\eqref{eq:PM_Darcy_Forch_1D}, we make use of the equation of state for an ideal gas $p(y) = \rho_f(y) R T_f(y)$, where~$p$ denotes the pressure. 
In 1D, the temperature system~\eqref{eq:temperature-system-full} can be solved independently from the mass-momentum system~\eqref{eq:mass-momentum-system}.

%
\textbf{Boundary conditions}.
In 1D, boundary conditions have to be imposed only on the reservoir~$\Gamma_{\text{R}}$ and on the interface~$\Gamma_{\text{Int}}$ because the side walls of the porous medium do not exist. However, in~\cite{RomMueller22_ijhmt} we show how the side walls can still be modeled for assembling a 2D or 3D porous medium solution from the individual 1D solutions.
\begin{itemize}
  \item On the reservoir $\Gamma_{\text{R}}$ at $y=0$:
  \begin{subequations}
  \begin{align}
    T_f(0) &= T_b,  \label{eq:T_f_reservoir_1D}\\
    T_s(0) &= T_b,  \label{eq:T_s_reservoir_1D}
  \end{align}
  \label{eq:bc_reservoir_1D}%
  \end{subequations}
  where $T_b$ is the constant backside temperature of the porous probe.
  
  \item On the interface $\Gamma_{\text{Int}}$ at $y=L$:
  \begin{subequations}
  \begin{align}
    \rho_f(L) &= \frac{p_{\text{HG}}}{R T_f(L)}, \label{eq:rho_f_int_1D}\\
    v(L) &= \frac{\dot m_c}{A_c} \frac{R T_f(L)}{p_{\text{HG}}}, \label{eq:velo_int_1D}\\
    \left(1-\varphi \right) \kappa_s T_s^{\prime}(L) &= q_{\text{HG}} - c_{p,f} \frac{\dot m_c}{A_c} \left(T_s(L) - T_f(L)\right), \label{eq:heat_balance_int_1D}\\
    T_f^{\prime}(L) &= \frac{\hv\,\Ac}{\cpf\,\mc} \left(T_{s}(L)-T_{f}(L)\right). \label{eq:Tf_int_1D}
  \end{align}
  \label{eq:bc_interface_1D}%
  \end{subequations}
  Here, $p_{\text{HG}}$ and $q_{\text{HG}}=\kappa_{\text{HG}} T_{\text{HG}}^{\prime}$ are the pressure and the heat flux, respectively, in the hot gas at the interface. For each 1D porous medium problem, both $p_{\text{HG}}$ and $q_{\text{HG}}$ are constant values such that the coordinate dependence of the hot gas values is omitted.
  
  The boundary condition~\eqref{eq:Tf_int_1D} for $T_f^{\prime}(L)$ is based on the simplification of the equation for the fluid temperature~\eqref{eq:PM_T_f_1D-full}. It is obtained by setting $y=L$ and neglecting the first term on the left-hand side, i.e., $-\varphi \kappa_f T_f^{\prime \prime}(y)$. This is a common simplification for cases in which $\kappa_f \ll \kappa_s$ such as transpiration cooling problems, see for instance~\cite{Wolfersdorf05,Colladay71}. Other works such as~\cite{Prokein20} use $T_f^{\prime}(L)=0$ instead, assuming that the convective heat flux from the hot gas side is fully absorbed by the solid in the porous medium. However, despite imposing $T_f^{\prime}(L)=0$ an adiabatic behavior of the fluid temperature at the interface is usually not observable in simulations where $\kappa_f \ll \kappa_s$. This is due to the strong coupling of the two temperature equations and fluid heat conduction being negligible. The influence of the different boundary conditions for $T_f^{\prime}(L)$ is investigated in Sect.~\ref{sec:numerical_results}.
\end{itemize}

\textbf{Solution procedure.}
The resulting 1D porous medium model becomes a  boundary value problem. Its solution requires to solve 
three subproblems for the temperatures $T_f$ and $T_s$, the density $\rho_f$ and the velocity $v$, respectively.
First of all, we solve 
the temperature system \eref{eq:temperature-system-full}
with boundary conditions
\eref{eq:T_f_reservoir_1D}, \eref{eq:T_s_reservoir_1D} at $y=0$ as well as \eref{eq:heat_balance_int_1D}, \eref{eq:Tf_int_1D} at $y=L$. 
Since these problems are linear, we can solve them explicitly.

The temperature  solution 
enters
the 
initial value problem~\eqref{eq:PM_conti_1D} for the density  with initial condition~\eref{eq:rho_f_int_1D} at $y=L$ as parameter. Under some assumptions, its solution can be verified to exist but is not known explicitly because the initial value problem is nonlinear. Instead, it will be approximated numerically applying some ODE solver. This is done by treating~\eqref{eq:PM_conti_1D} as a backward problem using
the boundary condition for $\rho_f(L)$ given in~\eqref{eq:rho_f_int_1D} and employing mass conservation \eref{eq:mass-cons-III}.

Finally, we may determine algebraically the solution of the velocity employing constant mass flow~\eref{eq:mass-cons-III}. 
In particular, \eqref{eq:velo_int_1D} holds at $y=L$.
Note that due to mass conservation  we alternatively may first solve the velocity and then determine the density.


\begin{remark}[Simplified temperature model]
As discussed above, if $\kappa_f \ll \kappa_s$, fluid heat conduction is negligible and the resulting simplified temperature system as used for our assembled-1D model in~\cite{RomMueller22_ijhmt} is given by
\begin{subequations}
\label{eq:temperature-system-simple}
\begin{align}
c_{p,f} \frac{\dot m_c}{A_c} T_{f}^{\prime}(y) &= h_v \left( T_s(y) - T_f(y) \right), 
\label{eq:PM_T_f_1D-simple}\\
(1-\varphi) \, \kappa_s \, T_{s}^{\prime \prime}(y) &= h_v \left( T_s(y) - T_f(y) \right), \label{eq:PM_T_s_1D-simple}
\end{align}
\label{eq:PM_system_1D}%
\end{subequations}
omitting the term $-\varphi \kappa_f T_{f}^{\prime \prime}(y)$ in~\eqref{eq:PM_T_f_1D-simple}.
Due to neglecting fluid heat conduction effects and, hence, the second-order derivative of the fluid temperature~$T_f$, the boundary condition~\eqref{eq:Tf_int_1D} for $T_f^{\prime}(L)$ is not needed for the simplified temperature model.

\end{remark}

In Sects.~\ref{sec:PM-full-temp-model-III} and \ref{sec:PM-mass-momentum}, we show under which 
assumptions the equations for the two temperatures including the term $-\varphi \kappa_f T_{f}^{\prime \prime}(y)$ and the one for the density have unique solutions. These assumptions 
ensure monotonicity of the temperatures, the density and the velocity.




\section{Temperature system with fluid heat conduction}
\label{sec:PM-full-temp-model-III}

For the investigation of the temperature system determined by~\eref{eq:temperature-system-full} and corresponding boundary conditions \eref{eq:T_f_reservoir_1D}, \eref{eq:T_s_reservoir_1D} at $y=0$ and \eref{eq:heat_balance_int_1D}, \eref{eq:Tf_int_1D} at $y=L$, we rewrite the model in a more canonical form. For this purpose, we introduce the derivatives of the temperatures as additional unknowns, i.e., $T_{f,1}:=T_f$, $T_{f,2}:=T_f'$, $T_{s,1}:=T_s$, $T_{s,2}:=T_s'$:
\begin{subequations}
\label{eq:ODE-full}
\begin{align}
\label{eq:ODE-full-c}
  T_{f,1}'(y)    &= T_{f,2}(y), \\[2mm]
\label{eq:ODE-full-d}
  T_{f,2}'(y)    &= -\frac{1}{\varphi \kappa_f} \left( \hv \left(T_{s,1}(y)-T_{f,1}(y)\right) -\cpf \frac{\mc}{\Ac} T_{f,2}(y) \right), \\[2mm]
\label{eq:ODE-full-e}
  T_{s,1}'(y)    &= T_{s,2}(y), \\[2mm]
\label{eq:ODE-full-f}
  T_{s,2}'(y)    &= \frac{1}{(1-\varphi)\kappa_s} \hv \left(T_{s,1}(y)-T_{f,1}(y)\right)
\end{align}
\end{subequations}
for $y\in(0,L)$ with boundary conditions
\begin{subequations}
\label{eq:ODE-full-bc}
\begin{align}
\label{eq:ODE-full-bc-0}
 & T_{f,1}(0) = T_b,\ T_{s,1}(0) = T_b,\\[2mm]
\label{eq:ODE-full-bc-L}
 & T_{f,2}(L) = \frac{\hv\,\Ac}{\cpf\,\mc} \left(T_{s,1}(L)-T_{f,1}(L)\right),\\[2mm]
 &  T_{s,2}(L) = \frac{1}{(1-\varphi)\kappa_s} \left(\cpf\frac{\mc}{\Ac} \left(T_{f,1}(L) - T_{s,1}(L)  \right)  + q_{\text{HG}}\right) .\nonumber
\end{align}
\end{subequations}
%

We solve this boundary value problem by means of a shooting procedure where in a first step we solve 
the linear homogeneous first order initial value problem
\begin{equation}
\label{eq:temp-system-full}
  \bz'(y) = \bA \bz(y),\quad y\in(0,L)
\end{equation}
with
\begin{equation}
  \bA :=
  \left(
  \begin{matrix}
    0  & 1 &  0 & 0\\
    a  & b & -a & 0\\
    0  & 0 &  0 & 1\\
    -c & 0 &  c & 0\\
  \end{matrix}
  \right)
,\quad
  \bz:=
  \left(
  \begin{matrix}
    T_{f,1} \\
    T_{f,2} \\
    T_{s,1} \\
    T_{s,2}
  \end{matrix}
  \right),
\end{equation}
positive matrix entries
\begin{equation}
\label{eq:pos-matrix-entries}
  a:= \frac{\hv}{\varphi \kappa_f},\
  b:= \frac{\cpf \mc }{\Ac\varphi\kappa_f},\
  c:= \frac{\hv}{(1-\varphi)\kappa_s}
\end{equation}
and initial conditions
\begin{equation}
\label{eq:temp-system-ic-full}
 \bz(0) = (T_f(0), s_1, T_s(0), s_2)^T = (T_b, s_1, T_b, s_2)^T \equiv\bz(0;s_1,s_2).
\end{equation}
Then the shooting parameters  $(s_1,s_2)$ are determined such that the solution of the initial value problem $\bz=\bz(\cdot;s_1,s_2)$ satisfies the boundary conditions
\begin{subequations}
\label{eq:bc-complex}%
\begin{align}
  &z_2(L;s_1,s_2) =
  \alpha \left( z_3(L;s_1,s_2)-z_1(L;s_1,s_2)\right),
 \\[2mm]
  &z_4(L;s_1,s_2) =
\beta \left( z_1(L;s_1,s_2)-z_3(L;s_1,s_2)\right)+\gamma
\end{align}
\end{subequations}
with
\begin{subequations}
\label{eq:bc-complex-coeff}%
\begin{align}
   & \alpha:= \frac{\hv\,\Ac}{\cpf\,\mc} = \frac{a}{b},\quad
     \beta:=\frac{1}{(1-\varphi)\kappa_s}\frac{\cpf\,\mc}{\Ac} = \frac{c\,b}{a},\\[2mm]
\label{eq:bc-complex-coeff-gamma}
   & \gamma:= \frac{1}{(1-\varphi)\kappa_s}  q_{\text{HG}}.
\end{align}
\end{subequations}

To determine the unique solution of the parameter-dependent initial value problem, we  determine the eigenvalues and eigenvectors of the matrix $\bA$. The characteristic polynomial is
\begin{align}
\label{eq:cubic-polynomial}
  \det(\bA - \lambda \bI) = \lambda\, p_3(\lambda)\quad \mbox{with}\quad
  p_3(\lambda) = \lambda^3 - b\,\lambda^2 - (a+c)\, \lambda + c\,b.
\end{align}
First of all, we investigate the roots of the cubic polynomial $p_3$.
\begin{proposition}[Properties of cubic polynomial]~\\[2mm]
\label{prop:cubicpolynomial}
There exist three real roots $\lambda_k$, $k=0,1,2$, of the polynomial $p_3$ with the following properties:
\begin{enumerate}
 \item the roots are explicitly given by
\begin{align}
  \label{eq:roots-cardano}
   &\lambda_k = 2 \sqrt{\frac{3(a+c)+b^2}{9}} \cos\left(\frac{\theta}{3}+\frac{2\,k\,\pi}{3} \right) + \frac{b}{3},\\[2mm]
  &\cos\,\theta = \frac{b}{2} \frac{2\,b^2 + 9\,(a+c) - 27\, c}{\left(3\,(a+c) +b^2 \right)^{3/2}}; \nonumber
\end{align}
 \item the roots can be estimated by
\begin{align}
\label{eq:root-interval}
  \lambda_0\in (\max(b,\sqrt{c}),+\infty),\quad
  \lambda_1\in (-\infty,-\sqrt{c}),\quad
  \lambda_2\in (0,\min(b,\sqrt{c}));
\end{align}
  \item the roots satisfy the relations
\begin{align}
\label{eq:roots-relations}
  b=\lambda_0+\lambda_1+\lambda_2,\
  a+c = -(\lambda_1\lambda_2+\lambda_0\lambda_1+\lambda_0\lambda_2),\
  c\,b = -\lambda_0 \lambda_1 \lambda_2 ;
\end{align}
\item the following estimates hold:
\begin{subequations}
\label{eq:estimate-roots}%
\begin{align}
\label{eq:estimate-roots-a}
   &c-\lambda_0^2<0,\quad c-\lambda_1^2<0,\quad   c-\lambda_2^2>0,\\[2mm]
\label{eq:estimate-roots-b}
   &  \lambda_2 < |\lambda_1|<\lambda_0.
\end{align}
\end{subequations}
\end{enumerate}
\end{proposition}
\begin{proof}
First of all, we check the existence of the three distinct real roots. For this purpose, we observe that the polynomial changes its sign in three disjoint intervals:
\begin{align}
\label{eq:intervals}
  & p_3(-\infty)=-\infty,\ p_3(-\sqrt{c})=a\,\sqrt{c}>0, \nonumber\\[2mm]
  & p_3(0)=c\,b>0, \  p_3(\sqrt{c})=-a\,\sqrt{c}<0, \\[2mm]
  & p_3(b)=-a\,b<0,\ p_3(+\infty)=+\infty. \nonumber
\end{align}
By continuity of the polynomial and the fundamental theorem of linear algebra, we conclude that there exist three distinct real roots in each of the intervals according to \eref{eq:root-interval}.
These roots can be determined by Cardano's formula
\[
  \lambda_k = 2 \varrho^{1/3} \cos\left(\frac{\theta}{3}+\frac{2\,k\,\pi}{3} \right) - \frac{r}{3},\
  \varrho=\sqrt{-\frac{p^3}{27}},\
  \cos\,\theta = -\frac{q}{2\,\varrho},\ k=0,1,2
\]
with
\[
  p=\frac{3\,s-r^2}{3},\
  q=\frac{2\,r^3}{27}-\frac{r\,s}{3} + t,\
  r=-b,\ s=-(a+c),\ t=c\,b.
\]
In terms of our coefficients, these are given by \eref{eq:roots-cardano}.
To determine the corresponding sign of the roots, we note that
by definition the angle $\theta\in [0,\pi]$. This implies for the angles $\theta_k:= (\theta+2\,k\,\pi)/3\in[2\,k,2\,k+1]\,\pi/3$ and, thus,
$\theta_0\in[0,\pi/3]$, $\theta_1 \in [2\,\pi/3,\pi]$ and $\theta_2\in[4\,\pi/3,5\,\pi/3]$.
For the cosine factors, we conclude $\cos(\theta_0)\in [0.5,1]$, $\cos(\theta_1)\in [-1,-0.5]$ and $\cos(\theta_2)\in [-0.5,0.5]$. This corresponds to \eref{eq:root-interval}.
\\
By the fundamental theorem of algebra, it holds
$p_3(\lambda)=(\lambda-\lambda_0)(\lambda-\lambda_1)(\lambda-\lambda_2)$.
By comparison of coefficients, the relations \eref{eq:roots-relations} must hold.\\
The estimate \eref{eq:estimate-roots-a} follows by \eref{eq:root-interval}. For the estimate \eref{eq:estimate-roots-b}, we conclude from $\lambda_0+\lambda_1+\lambda_2=b$ and $\lambda_0>b$ that $\lambda_1+\lambda_2<0$ holds. Because of
$\lambda_1<0$ and $\lambda_2>0$, this implies $\lambda_2<|\lambda_1|$ and, thus, $\lambda_0+\lambda_1>0$.
\end{proof}
\begin{remark}[Estimate of the roots]~\\[2mm]
\label{rem:signroots}
Note that the estimate $\lambda_0>b$ is  stronger than $\lambda_0>\sqrt{c}$ because $b$ is independent of $\hv$. Therefore,
the root $\lambda_0$ is always large  due to physical reasons. Then we may neglect terms $e^{-\lambda_0 L}$ later on.
\end{remark}
The fourth eigenvalue is $\lambda_3=0$ that is distinct to the others. Thus, there exist four linearly independent eigenvectors
\[
  \bv_k = c^{-1}\,(c-\lambda_k^2,(c-\lambda_k^2)\lambda_k, c, c\,\lambda_k)^T,\ k=0,1,2,\quad
  \bv_3 = (1,0,1,0)^T.
\]
Then the unique solution of the parameter-dependent initial value problem  is determined by
\begin{equation}
\label{eq:temp-solution-complex}
  \bz(y;s_1,s_2) =  \bW(y) \bW(0)^{-1} \bz(0;s_1,s_2) = \bV \bD(y) \bV^{-1} \bz(0;s_1,s_2)
\end{equation}
with Wronski matrix
\begin{align}
\label{eq:temp-system-Wronski-full}
 \bW(y) = \bV\,\bD(y),\quad
  \bV:= \left( \bv_0,\bv_1,\bv_2,\bv_3\right),\
  \bD(y) := \diag \left( e^{\lambda_0 y}, e^{\lambda_1 y},e^{\lambda_2 y},e^{\lambda_3 y}\right).
\end{align}
Here, $\bV$ denotes the matrix of right eigenvectors of the matrix  $\bA$
\[
\bV =
\left(
\begin{matrix}
 \frac{c-\lambda_0^2}{c} & \frac{c-\lambda_1^2}{c} & \frac{c-\lambda_2^2}{c} & 1 \\[2mm]
 \frac{(c-\lambda_0^2)\lambda_0}{c} & \frac{(c-\lambda_1^2)\lambda_1}{c} & \frac{(c-\lambda_2^2)\lambda_2}{c} & 0 \\[2mm]
 1 & 1 & 1 & 1 \\[2mm]
 \lambda_0 & \lambda_1 & \lambda_2 &  0
\end{matrix}
\right)
\]
with its inverse determined by
\[
  \bV^{-1} =
\left(
\begin{matrix}
  \sigma_0 & 0 & 0 & 0 \\
  0 & \sigma_1 & 0 & 0 \\
  0 & 0 & \sigma_2 & 0 \\
  0 & 0 & 0 & \sigma_3
\end{matrix}
\right) \times
\left(
\begin{matrix}
  -(b-\lambda_0) &  1 &   b-\lambda_0  &  \frac{b-\lambda_0}{\lambda_0} \\
    b-\lambda_1  & -1 & -(b-\lambda_1) & -\frac{b-\lambda_1}{\lambda_1} \\
  -(b-\lambda_2) &  1 &   b-\lambda_2  &  \frac{b-\lambda_2}{\lambda_2} \\
  -b^2 c & b c  & 0 & -(b-\lambda_0)(b-\lambda_2)(b-\lambda_1)
\end{matrix}
\right)
\]
and scaling factors
\[
\sigma_0:= \frac{\lambda_1\lambda_2}{b(\lambda_0 - \lambda_2)(\lambda_0 - \lambda_1)},\
\sigma_1:= \frac{\lambda_0\lambda_2}{b(\lambda_1 - \lambda_2)(\lambda_0 - \lambda_1)},\
\sigma_2:= \frac{\lambda_0\lambda_1}{b(\lambda_1 - \lambda_2)(\lambda_0 - \lambda_2)},\
\sigma_3:= -\frac{1}{b^2 c} .
\]
Thus, the solution of the initial value problem  determined by
\eref{eq:temp-system-full} --  \eref{eq:temp-system-ic-full} 
reads
\begin{subequations}
\label{eq:temp-solution-complex-3}%
\begin{align}
\label{eq:temp-solution-complex-3a}
 & z_1(y;s_1,s_2) = \sum_{k=0}^2 \frac{c-\lambda_k^2}{c} \sigma_k (-1)^k \left(s_1 +\frac{b-\lambda_k}{\lambda_k} s_2\right) e^{\lambda_k y} - \sigma_3 \left( b^2 c T_b -b c s_1 +s_2 \prod_{k=0}^2 (b-\lambda_k) \right)
\\[2mm]
 & z_2(y;s_1,s_2) = \sum_{k=0}^2 \frac{c-\lambda_k^2}{c} \lambda_k \sigma_k (-1)^k \left(s_1 +\frac{b-\lambda_k}{\lambda_k} s_2\right) e^{\lambda_k y} ,
\\[2mm]
\label{eq:temp-solution-complex-3c}
 & z_3(y;s_1,s_2) = \sum_{k=0}^2  \sigma_k (-1)^k \left(s_1 +\frac{b-\lambda_k}{\lambda_k} s_2\right) e^{\lambda_k y} - \sigma_3 \left( b^2 c T_b -b c s_1 +s_2 \prod_{k=0}^2 (b-\lambda_k) \right),
\\[2mm]
 & z_4(y;s_1,s_2) = \sum_{k=0}^2 \lambda_k \sigma_k (-1)^k \left(s_1 +\frac{b-\lambda_k}{\lambda_k} s_2\right) e^{\lambda_k y}.
\end{align}
\end{subequations}
The parameters $(s_1,s_2)$ are now chosen such that  the boundary conditions \eref{eq:bc-complex} hold at $y=L$ with $\bz(L;s_1,s_2)$ determined by \eref{eq:temp-solution-complex-3}.
This is equivalent to solving the linear $2\times2$  system
\begin{align}
\label{eq:s-system}
  \bC \bs = \bb
\end{align}
for $\bs = (s_1,s_2)^T$ with matrix $\bC$ and right-hand side $\bb$ determined by
\[
\bC:=
\left(
  \begin{matrix}
    \sum\limits_{k=0}^2 (-1)^k \lambda_k^2\, \frac{d_k^-}{d_k^+} \, e^{\lambda_k L} &
    \sum\limits_{k=0}^2 (-1)^k \lambda_k\, d_k^- \, e^{\lambda_k L} \\[2mm]
    \sum\limits_{k=0}^2 (-1)^k (a+b\,\lambda_k)\, d_k^- \, e^{\lambda_k L}  &
    \sum\limits_{k=0}^2 (-1)^k\, \frac{a+b\,\lambda_k}{\lambda_k} \,d_k^-\, d_k^+\, e^{\lambda_k L}
  \end{matrix}
\right)
,\quad
\bb:=
\left(
\begin{matrix}
  0\\
 \gamma \frac{d_0^+\,d_1^+\,d_2^+\,d_0^-\,d_1^-\,d_2^-}{b\, c}
\end{matrix}
\right)
\]
with
\begin{align}
\label{eq:dk_pm}
  d_0^\pm:= \lambda_1 \pm \lambda_2,\
  d_1^\pm:= \lambda_0 \pm \lambda_2,\
  d_2^\pm:= \lambda_0 \pm \lambda_1 .
\end{align}
Note  that by means of \eref{eq:roots-relations} it holds 
\begin{align}
\label{eq:relation-dp_k}
&d_k^+ = b-\lambda_k =\frac{a\,\lambda_k}{c-\lambda_k^2},\ k=0,1,2,
\\[2mm]
\label{eq:relation-dp012}
  & d_0^+\,d_1^+\,d_2^+ = -b\,a ,\\
  & d_0^-\,d_1^-\,d_2^- = - \frac{c\,d_k^-}{\lambda_k\,\sigma_k}.
\end{align}
Hence, the right-hand side of \eref{eq:s-system} simplifies to 
$\bb = (0,-\gamma\, a\,d_0^-\,d_1^-\,d_2^-/c)^T$.
Assuming that $\bC^0:=(\bc_1,\bc_2)$ is a regular matrix, the parameters $(s_1,s_2)$ are determined by
\begin{subequations}
\label{eq:s-param}%
\begin{align}
  &s_1 = \frac{\det(\bC^1)}{\det(\bC^0)} = \frac{c_{22}b_1-c_{12}b_2}{c_{11}c_{22}-c_{21}c_{12}},\
     \bC^1:=(\bb,\bc_2),\\[2mm]
  &s_2 = \frac{\det(\bC^2)}{\det(\bC^0)} = \frac{c_{11}b_2-c_{21}b_1}{c_{11}c_{22}-c_{21}c_{12}},\
   \bC^2:=(\bc_1,\bb)
\end{align}
\end{subequations}
with
\begin{subequations}
\label{eq:det}%
\begin{align}
  \det(\bC^0) = &
  \sum_{j,k=0}^2 C^0_{k,j}  e^{\lambda_k L}e^{\lambda_j L} ,
  C^0_{k,j}:= (-1)^{k+j} (a+b\,\lambda_j) \frac{\lambda_k}{\lambda_j}\frac{d_j^-\, d_k^-}{d_k^+}\left( \lambda_k\, d_j^+ - \lambda_j\, d_k^+ \right) \\[2mm]
  \det(\bC^1) = &
  \sum\limits_{k=0}^2  c^1_k \,e^{\lambda_k L} ,\
  c^1_k:= - \gamma \frac{d_0^+\,d_1^+\,d_2^+\,d_0^-\,d_1^-\,d_2^-}{b\, c}  (-1)^k \lambda_k\, d_k^- =
           (-1)^k \frac{\gamma \,a\,}{c} d_0^-\,d_1^-\,d_2^-\,\lambda_k\, d_k^-
,  \\[2mm]
  \det(\bC^2) = &
  \sum\limits_{k=0}^2 c^2_k\, e^{\lambda_k L} ,\
  c^2_k:= \gamma \frac{d_0^+\,d_1^+\,d_2^+\,d_0^-\,d_1^-\,d_2^-}{b\, c} (-1)^k \lambda_k^2\, \frac{d_k^-}{d_k^+} =
- (-1)^k \,\frac{\gamma\,a}{ c} \,d_0^-\,d_1^-\,d_2^-\,\lambda_k^2\, \frac{d_k^-}{d_k^+}  .
\end{align}
\end{subequations}
Computing the parameters $s_1$ and $s_2$ using the formulae \eref{eq:s-param} and \eref{eq:det} is not feasible because $a,b\gg 1$ and, thus, the positive eigenvalues are large, i.e., $\lambda_0,\lambda_2\gg 1$.
For a stable evaluation, it is recommended to 
factorize
the largest exponential term, $\exp((\lambda_0+\lambda_2)\,L)$, from the determinants.
For this purpose, we first note that
the number of summands in the determinant of the matrix $\bC^0$ can be reduced employing symmetry in the coefficients
\begin{align}
\label{eq:oC0_kj}
  \oC^0_{kj}:= C^0_{kj}+C^0_{jk} =
(-1)^{k+j+1} \frac{d_0^-\,d_1^- \, d_2^-}{a\,c} \left( a(a+\lambda_j\,\lambda_k)  + b^2 c\right) d_i^+\,d_i^-
\end{align}
for $k,i,j\in\{0,1,2\}$, $k<j$ and $i\ne k,j$.
In particular, we make use of the relations \eref{eq:roots-relations} and \eref{eq:cubic-polynomial}.
Factorizing
the largest exponential term from the determinants yields
\begin{align}
\label{eq:determinant}
 \det(\bC^0)   = e^{(\lambda_0+\lambda_2)L}\, \left( \oC^0_{02} + c^0_R\right),\
 \det(\bC^1)   = e^{(\lambda_0+\lambda_2)L} c^1_R,\
 \det(\bC^2)   = e^{(\lambda_0+\lambda_2)L} c^2_R
\end{align}
with coefficients
\begin{subequations}
\label{eq:det_coeff}%
\begin{align}
\label{eq:det_coeff_0}
 &c^0_R:=
\left( \oC^0_{01}\,e^{(\lambda_1-\lambda_2)L} +
                               \oC^0_{12}\,e^{(\lambda_1-\lambda_0)L}  \right),  \\[2mm]
\label{eq:det_coeff_1}
 &c^1_R:=
\left( c^1_0\,e^{-\lambda_2 L} +
                               c^1_1\,e^{(\lambda_1-\lambda_0-\lambda_2) L} +
                               c^1_2\,e^{-\lambda_0 L}  \right),  \\[2mm]
\label{eq:det_coeff_2}
 &c^2_R:=
\left( c^2_0\,e^{-\lambda_2 L} +
                               c^2_1\,e^{(\lambda_1-\lambda_0-\lambda_2) L} +
                               c^2_2\,e^{-\lambda_0 L}  \right) .
\end{align}
\end{subequations}
Finally, we can rewrite the parameters $s_1$ and $s_2$ as
\begin{align}
\label{eq:param-stable}
  s_1 = \frac{c^1_R}{\oC^0_{02} + c^0_R},\quad
  s_2 = \frac{c^2_R}{\oC^0_{02} + c^0_R} .
\end{align}
So far we assume regularity and, thus, the existence of the parameters $(s_1,s_2)$. We now verify regularity of the matrix $\bC^0$.
\begin{lemma}[Regularity of matrix $\bC^0$]~\\[2mm]
\label{lem:reg-determinant}
The determinant of the matrix $\bC^0$ is positive, i.e., $\bC^0$ is  regular, for all $L\ge 0$. In particular, the linear system of equations \eref{eq:s-system} has a unique solution determined by \eref{eq:s-param} or, equivalently, \eref{eq:param-stable}.
\end{lemma}

\begin{proof}
According to \eref{eq:determinant}, \eref{eq:oC0_kj} and \eref{eq:det_coeff_0}, the determinant of the matrix $\bC^0$ can be written as
\begin{align}
\label{eq:detC0-help}
   \det(\bC^0) =  
- e^{(\lambda_0+\lambda_2) L}\, \frac{d_0^-\,d_1^-\,d_2^-}{a\,c}
\left( 
\ooC^0_{02} + \ooC^0_{01}\,e^{(\lambda_1-\lambda_2) L} +  \ooC^0_{12}\,e^{(\lambda_1-\lambda_0) L}
\right)
\end{align}
with coefficients
\begin{align}
\label{eq:ooC_kj}
  \ooC_{kj}^0 := - \frac{a\,c}{d_0^-\,d_1^-\,d_2^-}\, \oC_{kj}^0 = 
  (-1)^{k+j} \left( a(a+\lambda_j\,\lambda_k) + b^2 c\right) d_i^+\,d_i^-,\
  i,k,j\in\{0,1,2\},\ k\ne j\ne i.
\end{align}
To estimate the sign of the determinant, we will employ the following estimates
\begin{subequations}
\label{eq:estimates}%
\begin{align}
\label{eq:estimates-dm_k}
  &d_0^- = \lambda_1 - \lambda_2 <0,\
   d_1^- = \lambda_0 - \lambda_2 >0,\
   d_2^- = \lambda_0 - \lambda_1 >0,\\
  &d_0^+ = \lambda_1 + \lambda_2 <0,\
   d_1^+ = \lambda_0 + \lambda_2 >0,\
   d_2^+ = \lambda_0 + \lambda_1 >0,\\
 & \lambda_0 >0,\
   \lambda_1 <0 ,\
   \lambda_2 >0
\end{align}
\end{subequations}
that hold because of Proposition \ref{prop:cubicpolynomial}. Then the prefactor in \eref{eq:detC0-help} is positive. Obviously, the coefficient $\ooC_{02}^0$ is also positive. To determine the sign of the determinant, we now distinguish four cases for the the sign of the remaining two coefficients to estimate the sum in the parentheses.\\
\textbf{Case 1}: If both $\ooC_{01}^0$ and $\ooC_{12}^0$ are positive, then all the terms in the parentheses of  \eref{eq:detC0-help} are positive.\\[0.1cm]
\textbf{Case 2}: If both $\ooC_{01}^0$ and $\ooC_{12}^0$ are negative, then we deduce by means of \eref{eq:estimates}
\begin{align*}
\ooC^0_{02} + \ooC^0_{01}\,e^{(\lambda_1-\lambda_2) L} +  \ooC^0_{12}\,e^{(\lambda_1-\lambda_0) L} \ge
\ooC^0_{02} + \ooC^0_{01} +  \ooC^0_{12} = 
  - a\,b\,   d_0^- \, d_1^- \, d_2^- > 0 .
\end{align*}
\textbf{Case 3}: If  $\ooC_{01}^0$ is positive and $\ooC_{12}^0$ is negative, then the following estimate holds
\begin{align*}
\ooC^0_{02} + e^{(\lambda_1-\lambda_2) L}\, \left( \ooC^0_{01} +  \ooC^0_{12}\,e^{(\lambda_2-\lambda_0) L} \right) \ge
e^{(\lambda_1-\lambda_2) L}\, \left(\ooC^0_{02} +  \ooC^0_{01} +  \ooC^0_{12} \right) \ge 0.
\end{align*}
\textbf{Case 4}: Finally, assuming that $\ooC_{01}^0$ is negative and $\ooC_{12}^0$ is positive, the definition \eref{eq:ooC_kj} of the coefficients implies $\lambda_2>\lambda_0$. This contradicts $\eref{eq:estimates-dm_k}$. Thus, this case cannot occur.\\
Finally, we conclude that the determinant of the matrix $\bC^0$ is positive for all $L>0$.
\end{proof}

Since now the parameters can be computed,  we  consider next the evaluation of the temperatures
$T_f(y)\equiv z_1(y;s_1,s_2)$ and $T_s(y)\equiv z_3(y;s_1,s_2)$. Again, the formulae
\eref{eq:temp-solution-complex-3a} and \eref{eq:temp-solution-complex-3c} are not feasible because of the large exponential factors. Therefore, we proceed as before and 
factorize
the largest factor, where we split $T_{f,s}$ in a non-constant and a constant part.
We do this separately for the terms
\begin{subequations}
\label{eq:temp-complex-split}%
\begin{align}
\label{eq:temp-complex-split-a}
&T_{f,s}(y) = \oT_{f,s}(y) + \oT :=
             \sum_{k=0}^2 c_{f,s}^k b_k e^{\lambda_k y}
             -\sigma_3\left( c\,b^2\, T_b - c\, b\, s_1 + s_2 \prod_{k=0}^2 (b-\lambda_k) \right),\\[2mm]
\label{eq:temp-complex-split-b}
&c_{f}^k:=(-1)^k \frac{\sigma_k}{\lambda_k}\frac{c-\lambda_k^2}{c},\
 c_{s}^k:=(-1)^k \frac{\sigma_k}{\lambda_k},\
 b_k:= \lambda_k\,s_1 + (b-\lambda_k)\,s_2 .
\end{align}
\end{subequations}
Inserting \eref{eq:param-stable} into the term $\oT$ and employing the splitting $\ooC_{02}^k+c_R^k$ into large and small contributions, we obtain a similar splitting
\begin{align}
  & \oT := \frac{1}{1+\oc_R^0}\, \left( \oT_1 +  \oT_2 \right), \nonumber \\[2mm]
  & \oT_1:=  -\sigma_3\,b^2\,c\,T_b = T_b, \nonumber\\[2mm]
  & \oT_2:=  -\sigma_3\,\frac{1}{\oC^0_{02}}\, \left( b^2\,c\,T_b\, c^0_R - b\,c\,c_R^1 + c_R^2 (b-\lambda_0)(b-\lambda_1) (b-\lambda_2) \right)
\nonumber\\[2mm]
&\hphantom{\oT_2}
 = \oT_{01} e^{(\lambda_1-\lambda_2) L} + \oT_{12} e^{(\lambda_1-\lambda_0) L} +
  \ooT_0\, e^{-\lambda_2 L} + \ooT_1\, e^{(\lambda_1-\lambda_0-\lambda_2) L} + \ooT_2\, e^{-\lambda_0 L}, \nonumber\\[2mm]
&\oT_{kj}:=  -\sigma_3\,\frac{\oC^0_{kj}}{\oC^0_{02}}\, \,b^2\, c \, T_b
= \frac{\oC^0_{kj}}{\oC^0_{02}}  \, T_b = (-1)^{k+j} \frac{d_i^-\,d_i^+}{d_1^-\,d_1^+} \frac{a\,(a+\lambda_j\,\lambda_k)+b^2\,c}{a\,(a+\lambda_0\,\lambda_2)+b^2\,c}\,T_b
, \nonumber \\[2mm]
&\ooT_k:= \sigma_3\,\frac{1}{\oC^0_{02}}\,  \left(  b\,c\,\,c^1_k - (b-\lambda_0)(b-\lambda_1)(b-\lambda_2)\,c_k^2 \right) =
(-1)^{k}\,\frac{\gamma\, a^2}{c\,b}\,\frac{1}{a\,(a+\lambda_0\,\lambda_2)+c\,b^2}\,\frac{d_k^-\,\lambda_k^3}{d_1^-\,d_1^+},\nonumber \\[2mm]
&\oc_R^0 :=\frac{c_R^0}{\oC^0_{02}} = 
\oc^0_{01}\,e^{(\lambda_1-\lambda_2) L} +
                               \oc^0_{12}\,e^{(\lambda_1-\lambda_0) L}, \nonumber\\[2mm]
&\oc^0_{kj}:= \frac{\oC^0_{kj}}{\oC^0_{02}} = 
(-1)^{k+j} \frac{d_i^-\,d_i^+}{d_1^-\,d_1^+} \frac{a\,(a+\lambda_j\,\lambda_k)+b^2\,c}{a\,(a+\lambda_0\,\lambda_2)+b^2\,c} = \frac{\oT_{kj}}{T_b}. \nonumber
\end{align}
Next we consider $\oT_{f,s}$. We start with separating small and large contributions in the coefficients $b_k$:
\begin{align}
 & b_k := \frac{ b_k^2}{1+\oc_R^0}, \nonumber \\[2mm]
  & b_k^2 := \frac{1}{\oC^0_{02}}\, \left(  \lambda_k\,c_R^1+(b-\lambda_k)\,c_R^2 \right) \nonumber
=
  \ob_k^0\, e^{-\lambda_2 L} + \ob_k^1\, e^{(\lambda_1-\lambda_0-\lambda_2) L} + \ob_k^2\, e^{-\lambda_0 L},\nonumber \\[2mm]
& \ob_k^l := \frac{1}{\oC^0_{02}}\, \left( \lambda_k\,c^1_l + (b-\lambda_k)\,c^2_l\right) =
(-1)^l \,  \frac{\gamma\,b\, a^2}{a\,(a+\lambda_0\lambda_2)+c\,b^2} \frac{1}{d_1^-\,d_1^+}
\frac{\lambda_l\,d_l^-}{d_l^+}  ( \lambda_l -  \lambda_k).
\nonumber 
\end{align}
Incorporating the splitting of $b_k$ in $\oT_{f,s}$, we obtain
\begin{align}
  &\oT_{f,s}(y) := \frac{1}{1+\oc_R^0}\,  \oT_{f,s}^2 (y) ,
\nonumber\\[2mm]
 & \oT_{f,s}^2(y) := \sum_{k=0}^2 c_{f,s}^k\,b_k^2\, e^{\lambda_k y} =
\sum_{k=0}^2 \oT^{2,k}_{f,s}(y), \nonumber \\[2mm]
  & \oT^{2,0}_{f,s}(y) :=
   \sum_{k=0}^2 c_{f,s}^k\,\ob_k^0\, e^{-\lambda_2 L+\lambda_k y}
  = \sum_{k=1,2} c_{f,s}^k\,\ob_k^0\, e^{-\lambda_2 L+\lambda_k y}
  = \sum_{k=1,2} \oc_{f,s}^{0,k}\, e^{-\lambda_2 L+\lambda_k y}, \nonumber \\[2mm]
  & \oT^{2,1}_{f,s}(y) :=
   \sum_{k=0}^2 c_{f,s}^k\,\ob_k^1\, e^{(\lambda_1-\lambda_0-\lambda_2) L+\lambda_k y}
  = \sum_{k=0,2} c_{f,s}^k\,\ob_k^1\, e^{(\lambda_1-\lambda_0-\lambda_2) L+\lambda_k y}  
 = \sum_{k=0,2} \oc_{f,s}^{1,k}\, e^{(\lambda_1-\lambda_0-\lambda_2) L+\lambda_k y}  
, \nonumber \\[2mm]
  & \oT^{2,2}_{f,s}(y) :=
   \sum_{k=0}^2 c_{f,s}^k\,\ob_k^2\, e^{-\lambda_0 L+\lambda_k y}
  = \sum_{k=0,1} c_{f,s}^k\,\ob_k^2\, e^{-\lambda_0 L+\lambda_k y}
  = \sum_{k=0,1} \oc_{f,s}^{2,k}\, e^{-\lambda_0 L+\lambda_k y}, \nonumber \\[2mm]
  & \oc_f^{l,k}:=c_f^k\,\ob_k^l = 
(-1)^{l+k+1} \frac{\gamma\,b\,a^2}{a\,(a+\lambda_0\lambda_2)+c\,b^2} \, 
 \frac{1}{d_0^-\,d_1^-\,d_2^-}  \,  \frac{1}{ d_1^+\, d_1^-}\,  
\frac{d_l^-\, d_k^-}{d_l^+}\,\frac{\lambda_l}{\lambda_k^2}\,(c-\lambda_k^2)\,(\lambda_l-\lambda_k),\  
\nonumber\\[2mm]
  & \oc_s^{l,k}:=c_s^k\,\ob_k^l = 
(-1)^{l+k+1} \frac{\gamma\,b\,a^2}{a\,(a+\lambda_0\lambda_2)+c\,b^2}\, 
 \frac{1 }{d_0^-\,d_1^-\,d_2^-} \, \frac{1}{ d_1^+\, d_1^-}\, \frac{\lambda_l}{\lambda_k^2} \,
\frac{d_k^-\,d_l^-}{d_l^+}\,c\,(\lambda_l-\lambda_k).
\nonumber
\end{align}
Combining the above formulae, we obtain the following representation for the temperatures
\begin{align}
\label{eq:Tfs}
&T_{f,s}(y) = 
T_b - \frac{\gamma\,a}{c}
\frac{d_0^-\,d_1^-\,d_2^-}{\oC^0_{02}+c_R^0}
\left( 
 \ooc^{0,1}_{f,s} e^{-\lambda_2 L+\lambda_1 y} +
 \ooc^{0,2}_{f,s} e^{-\lambda_2 L+\lambda_2 y} +
 \ooc^{1,0}_{f,s} e^{(\lambda_1-\lambda_0-\lambda_2) L+\lambda_0 y} + \right.
\nonumber\\[2mm]
&\hspace*{45mm}
 \ooc^{1,2}_{f,s} e^{(\lambda_1-\lambda_0-\lambda_2) L+\lambda_2 y} +
 \ooc^{2,0}_{f,s} e^{-\lambda_0 L+\lambda_0 y} +
 \ooc^{2,1}_{f,s} e^{-\lambda_0 L+\lambda_1 y} +
 \nonumber\\[2mm]
&\hspace*{44mm}
\left. 
\frac{1}{c\, b}
\left(
  \oooT_0\, e^{-\lambda_2 L} + \oooT_1\, e^{(\lambda_1-\lambda_0-\lambda_2) L} + \oooT_2\, e^{-\lambda_0 L}
\right)
\right)
\end{align}
with coefficients for $l,k,i\in \{0,1,2\}$, $l\ne k\ne i$,
\begin{subequations}
\begin{align}
\ooc^{l,k}_{f}&:= -\frac{c}{\gamma\, a} \oC^0_{02}\, \frac{1}{d_0^-\,d_1^-\,d_2^-}\,\oc^{l,k}_{f} =
(-1)^{l+k} \, \sign(k-l)\,\frac{d_i^+\,\lambda_l}{\lambda_k},\\[2mm]
\ooc^{l,k}_{s}&:= -\frac{c}{\gamma\, a} \oC^0_{02}\, \frac{1}{d_0^-\,d_1^-\,d_2^-}\,\oc^{l,k}_{s} =
(-1)^{l+k}   \, \sign(k-l)\,\frac{\lambda_i\,\lambda_l^2}{\lambda_k\,d_l^+},\\[2mm]
\oooT_k &:= -\frac{c}{\gamma\,b\, a} \oC^0_{02}\, \frac{c\,b^2}{d_0^-\,d_1^-\,d_2^-} \, \ooT_k =
(-1)^k\,\lambda_k^3\,d_k^-
\end{align}
\end{subequations}
and $\oC^0_{02}$ , $c_R^0$, $d_k^\pm$ defined by \eref{eq:det_coeff_0}, \eref{eq:oC0_kj} and \eref{eq:dk_pm}, respectively. Furthermore, the parameters $a$, $b$, $c$ and $\gamma$ are defined in \eref{eq:pos-matrix-entries} and \eref{eq:bc-complex-coeff-gamma}. Finally, the eigenvalues $\lambda_k$ are determined by \eref{eq:roots-cardano}.

To derive  properties of the solution of the boundary value problem 
\eref{eq:temp-system-full} -- \eref{eq:bc-complex},
the monotonicity of the temperatures is of key importance. 
For this purpose, we consider the derivative of the temperatures given by
\begin{align}
\label{eq:Tfs_prime}
&T'_{f,s}(y) = 
 - \frac{\gamma\,a}{c}
\frac{d_0^-\,d_1^-\,d_2^-}{\oC^0_{02}+c_R^0}
\left( 
 \ood^{0,1}_{f,s} e^{-\lambda_2 L+\lambda_1 y} +
 \ood^{0,2}_{f,s} e^{-\lambda_2 L+\lambda_2 y} +
 \ood^{1,0}_{f,s} e^{(\lambda_1-\lambda_0-\lambda_2) L+\lambda_0 y} + 
\right.
\nonumber\\[2mm]
&\hspace*{43mm}
\left. 
 \ood^{1,2}_{f,s} e^{(\lambda_1-\lambda_0-\lambda_2) L+\lambda_2 y} +
 \ood^{2,0}_{f,s} e^{-\lambda_0 L+\lambda_0 y} +
 \ood^{2,1}_{f,s} e^{-\lambda_0 L+\lambda_1 y} 
\right) 
\end{align}
with coefficients $\ood^{l,k}_{f,s}:= \ooc^{l,k}_{f,s}\, \lambda_k $. 

\begin{lemma}[Estimates for the derivatives of the temperatures]~\\[2mm]
\label{lem:temp-der}
Let be $L>0$.
Then the derivatives of the temperatures can be estimated by
\begin{align}
\label{eq:estimate_Tfs_prime}
T'_f(y) \ge 
 \frac{\gamma\,a}{c^2} \frac{(d_0^-\,d_1^-\,d_2^-)^2}{\oC^0_{02}+c_R^0}
 e^{-\lambda_2 L+\lambda_2 y}  \ge 0,\
T'_s(y) \ge 
 \frac{\gamma\,a}{c^2} \frac{(d_0^-\,d_1^-\,d_2^-)^2}{\oC^0_{02}+c_R^0}
 e^{(\lambda_1-\lambda_0-\lambda_2) L+\lambda_2 y}  \ge 0 
\end{align}
for $0\le y \le L$.
In particular, the derivatives vanish if and only if $\gamma=0$.
\end{lemma}

\begin{proof}
Since by Lemma \ref{lem:reg-determinant} the determinant of the matrix $\bC^0$ is positive,
the term $\oC^0_{02}+c_R^0$ is positive as well. By means of the estimates \eref{eq:estimates}, the prefactor in \eref{eq:Tfs_prime} is positive and it holds
\begin{align*}
  & \ood^{0,1}_{f}<0,\ \ood^{0,2}_{f}>0,\ \ood^{1,0}_{f}<0,\ \ood^{1,2}_{f}<0,\ 
    \ood^{2,0}_{f}<0,\ \ood^{2,1}_{f}<0,   \\[2mm]
 & \ood^{0,1}_{s}>0,\ \ood^{0,2}_{s}>0,\ \ood^{1,0}_{s}>0,\ \ood^{1,2}_{s}<0,\ 
    \ood^{2,0}_{s}>0,\ \ood^{2,1}_{s}>0.
\end{align*}
For the sum in the parenthesis in \eref{eq:Tfs_prime}, we consider the fluid and the solid case separately.
According to the eigenvalues, the term $e^{-\lambda_2 L+\lambda_2 y} $ decays much slower than the other exponential terms. For the fluid temperature, we factorize this term and estimate the remaining part by
\begin{align*}
&
 \ood^{0,1}_{f,s} e^{(\lambda_1-\lambda_2) y} +
 \ood^{0,2}_{f,s}  +
 \ood^{1,0}_{f,s} e^{(\lambda_1-\lambda_0) L+(\lambda_0-\lambda_2) y} + 
\nonumber\\[2mm]
&
 \ood^{1,2}_{f,s} e^{(\lambda_1-\lambda_0) L} +
 \ood^{2,0}_{f,s} e^{(\lambda_2-\lambda_0) (L-y)} +
 \ood^{2,1}_{f,s} e^{(\lambda_2-\lambda_0) L+(\lambda_1-\lambda_2) y} 
 \ge\\[2mm]
&\ood^{0,1}_f  + \ood^{0,2}_f + \ood^{1,0}_f + \ood^{1,2}_f  + \ood^{2,0}_f + \ood^{2,1}_f =
\frac{d_0^-\,d_1^-\,d_2^-\,(\lambda_0+\lambda_1+\lambda_2)}{\lambda_0\,\lambda_1\,\lambda_2} =
-\frac{d_0^-\,d_1^-\,d_2^-}{c} .
\end{align*}
Here we use that for $0\le y\le L$ all exponents are non-positive and, thus, the exponential functions can be estimated by 1 from above, whereas the factors of these exponential terms are all negative.\\
For the solid temperature, we proceed similarly. According to the eigenvalues, the term $e^{(\lambda_1-\lambda_0-\lambda_2) L+\lambda_2 y}$ decays much faster than the other exponential terms. Again, we factorize this term from the parenthesis term  in \eref{eq:Tfs_prime} and estimate the remaining part by
\begin{align*}&
 \ood^{0,1}_{f,s} e^{(\lambda_0-\lambda_1) L+(\lambda_1-\lambda_2) y} +
 \ood^{0,2}_{f,s} e^{(\lambda_0-\lambda_1) L} +
 \ood^{1,0}_{f,s} e^{(\lambda_0-\lambda_2) L} + 
\nonumber\\[2mm]
& \ood^{1,2}_{f,s} +
 \ood^{2,0}_{f,s} e^{(\lambda_2-\lambda_1) L+(\lambda_0-\lambda_2) y} +
 \ood^{2,1}_{f,s} e^{(\lambda_2-\lambda_1) (L-y)} 
\ge\\[2mm]
&\ood^{0,1}_s  + \ood^{0,2}_s + \ood^{1,0}_s + \ood^{1,2}_s  + \ood^{2,0}_s + \ood^{2,1}_s =
\frac{d_0^-\,d_1^-\,d_2^-\,(\lambda_0+\lambda_1+\lambda_2)}{\lambda_0\,\lambda_1\,\lambda_2} =
-\frac{d_0^-\,d_1^-\,d_2^-}{c} .
\end{align*}
Here we use that for $0\le y\le L$ all exponents are non-negative and, thus, the exponential functions can be estimated by 1 from below, whereas the factors of these exponential terms are all positive.\\
Finally, combining the above findings we conclude with the estimates \eref{eq:estimate_Tfs_prime}.
Note that equality only holds if $\gamma=0$.
\end{proof}

%

Next we investigate the temperature difference $T_s-T_f$ that according to \eref{eq:Tfs} is given by
\begin{align}
\label{eq:diff_Ts_Tf}
&T_s(y)- T_f(y) = 
 - \frac{\gamma\,a}{c}
\frac{d_0^-\,d_1^-\,d_2^-}{\oC^0_{02}+c_R^0}
\left( 
 \cc^{0,1} e^{-\lambda_2 L+\lambda_1 y} +
 \cc^{0,2} e^{-\lambda_2 L+\lambda_2 y} +
 \cc^{1,0} e^{(\lambda_1-\lambda_0-\lambda_2) L+\lambda_0 y} + \right.
\nonumber\\[2mm]
&\hspace*{51mm}
\left. 
 \cc^{1,2} e^{(\lambda_1-\lambda_0-\lambda_2) L+\lambda_2 y} +
 \cc^{2,0} e^{-\lambda_0 L+\lambda_0 y} +
 \cc^{2,1} e^{-\lambda_0 L+\lambda_1 y} 
\right) 
\end{align}
with coefficients 
\begin{subequations}
\begin{align}
\cc^{l,k}&:= \ooc^{l,k}_s - \ooc^{l,k}_f =
(-1)^{l+k}   \, \sign(k-l)\,\frac{\lambda_l\,}{\lambda_k\,d_l^+}
\left( \lambda_i\,\lambda_l - d_i^+\,d_l^+ \right) =
(-1)^{l+k+1}\,\sign(k-l) \, \frac{\lambda_l\,b}{d_l^+} 
\end{align}
\end{subequations}
for $l,k,i\in \{0,1,2\}$, $l\ne k\ne i$.

\begin{lemma}[Estimates of the temperature difference]~\\[2mm]
\label{lem:temp-diff}
Let be $L>0$.
Then the difference of the temperatures can be estimated by
\begin{align}
\label{eq:estimate_Ts_Tf}
T_s(y) \ge T_f(y) ,\quad y\in[0,L].
\end{align}
In particular, the temperatures coincide if and only if $\gamma=0$.
\end{lemma}

\begin{proof}
Since by Lemma \ref{lem:reg-determinant} 
the determinant of the matrix $\bC^0$ is positive,
the term $\oC^0_{02}+c_R^0$ is positive as well. By means of the estimates \eref{eq:estimates}, the prefactor in \eref{eq:diff_Ts_Tf} is positive and it holds
\begin{align*}
   \cc^{0,2} = -\cc^{0,1} >0,\ \cc^{1,0} = - \cc^{1,2} >0,\ \cc^{2,0} = - \cc^{2,1} >0.
\end{align*}
Then the sum in the parenthesis of \eref{eq:diff_Ts_Tf} can be estimated by
\begin{align*}
& \cc^{0,2} \left( e^{-\lambda_2 L + \lambda_2 y} - e^{-\lambda_2 L + \lambda_1 y} \right) +
  \cc^{1,0} \left( e^{(\lambda_1 - \lambda_0 -\lambda_2) L + \lambda_0 y} - 
                   e^{(\lambda_1 - \lambda_0 -\lambda_2) L + \lambda_2 y} \right) +
\nonumber\\[2mm]
&
   \cc^{2,0} \left( e^{-\lambda_0 L + \lambda_0 y} - e^{-\lambda_0 L + \lambda_1 y} \right)
 \ge 0 ,
\end{align*}
where we employ positivity of the coefficients and \eref{eq:estimates-dm_k}.
Finally, combining the above findings we conclude with the estimate \eref{eq:estimate_Ts_Tf}.
Note that due to \eref{eq:diff_Ts_Tf} equality only holds if $\gamma=0$.
\end{proof}


By means of Lemma \ref{lem:reg-determinant}, \ref{lem:temp-der} and \ref{lem:temp-diff}, we may now conclude on the following properties of the temperatures.

%
\begin{proposition}[Properties of the temperature system with fluid heat conduction]~\\[2mm]
\label{prop:A-temp-full-III}
The following properties hold true for the temperature system determined by the boundary value problem 
\eref{eq:temp-system-full} -- \eref{eq:bc-complex} and arbitrary $0<L<\infty$:
\begin{enumerate}
  \item There exists a unique solution of the boundary value problem 
        \eref{eq:temp-system-full} -- \eref{eq:bc-complex}
        determined by \eref{eq:temp-solution-complex-3} with $(s_1,s_2)$ the solution of
        \eref{eq:s-system}.
  \item If  $\gamma>0$, i.e., 
        \begin{align}
          \label{eq:A-assumption-temp-III} 
           q_{\text{HG}}>0,
        \end{align}
        then it holds:
        \begin{enumerate}
          \item The temperature difference $T_s-T_f$ is positive.
         \item The coolant temperature $T_f$ is monotonically increasing and
                \begin{align}
                  \label{eq:tfgetc}
                   T_f(y) \ge T_b, \quad y\in[0,L] .
                \end{align}
          \item The solid temperature $T_s$ is monotonically increasing and
                \begin{align}
                  \label{eq:tsgetc}
                  T_s(y) \ge T_b, \quad y\in[0,L] .
                \end{align}
        \end{enumerate}
  \item If  $\gamma=0$, i.e., $q_{\text{HG}}=0$, then the temperatures are constant and it holds:
                \begin{align}
                  T_f(y) = T_s(y) = T_b, \quad y\in[0,L] .
                \end{align}
\end{enumerate}
\end{proposition}

\begin{proof}
The existence of a unique solution is ensured by the regularity of the matrix $\bC^0$ according to 
Lemma \ref{lem:reg-determinant} 
providing us with unique parameters $(s_1,s_2)$ determined by \eref{eq:param-stable}.\\
For $\gamma>0$, we conclude by Lemma \ref{lem:temp-der}  that the temperatures $T_f$ and $T_s$ are strictly monotonically increasing and, thus, by the initial conditions \eref{eq:temp-system-ic-full} are larger than or equal to $T_b$. Furthermore, Lemma~\ref{lem:temp-diff} verifies that $T_s$ is always larger than or equal to $T_f$.\\
Finally, for $\gamma=0$, we conclude from \eref{eq:Tfs} that both temperatures coincide with the initial value $T_b$.
\end{proof}

\begin{remark}[Simplified temperature system]
In \cite{RomMueller22_ijhmt}, Appendix A,  the solution to the simplified temperature system determined by \eref{eq:temperature-system-simple} and corresponding boundary conditions \eref{eq:T_f_reservoir_1D}, \eref{eq:T_s_reservoir_1D} at $y=0$ and \eref{eq:heat_balance_int_1D} at $y=L$ is determined as
  \begin{subequations}
  \label{eq:A-simplified-III-solution-temp}%
\begin{align}
  \label{eq:A-simplified-III-solution-temp-Tf}
  & T_f(y) = T_b - \oor\,\ooa\, \frac{\lambda_- - \lambda_+ + \lambda_+ e^{\lambda_- y} - \lambda_- e^{\lambda_+ y} }{\lambda_+^2 e^{\lambda_- L}  - \lambda_-^2 e^{\lambda_+ L}}    ,    \\
  \label{eq:A-simplified-III-solution-temp-Ts}
  & T_s(y) = T_b + \oor \frac{\lambda_+^2 e^{\lambda_- y}  - \lambda_-^2 e^{\lambda_+ y}  -  (\lambda_- - \lambda_+) \ooa  }{\lambda_+^2 e^{\lambda_- L}  - \lambda_-^2 e^{\lambda_+ L}}    
%
\end{align}
\end{subequations}
with
\begin{align}
  \label{eq:abr-III}
  & \ooa:= \frac{\hv\Ac}{\cpf\,\mc},\
  \oob:= \frac{\hv}{(1-\varphi)\kappa_s},\\
  \label{eq:A-eigenvalues-simplified-III}
  &\lambda_\pm = \frac{1}{2} \left( -\ooa \pm \sqrt{\ooa^2 + 4\,  \oob} \right),\\
  \label{eq:A-simplified-III-r}
  &\oor := \frac{\ooa}{\hv}q_{\text{HG}}.
\end{align}
For these temperatures, the same properties  hold true as in Proposition \ref{prop:A-temp-full-III}.
\end{remark}



\section{Mass-momentum system}
\label{sec:PM-mass-momentum}

For the investigation of the mass-momentum system determined by~\eref{eq:mass-momentum-system} and corresponding initial values \eref{eq:rho_f_int_1D} and \eref{eq:velo_int_1D}, it is useful to rewrite the model in a more canonical form. Employing constant mass flow~\eref{eq:mass-cons-III} and conservation of mass, the mass-momentum equations \eref{eq:PM_conti_1D} and \eref{eq:PM_Darcy_Forch_1D}
can be rewritten and we finally obtain the system of ordinary differential equations (ODEs)
%
%
%
\begin{subequations}
\label{eq:ODE-simplified-III}%
\begin{align}
\label{eq:ODE-simplified-III-a}
  \rho_f'(y) &=  \oN((\rho_f,v,T_f,T_s)(y))\, \rho_f(y), \\[2mm]
\label{eq:ODE-simplified-III-b}
  v'(y)      &= -\oN((\rho_f,v,T_f,T_s)(y))\, v(y)
\end{align}
\end{subequations}
for $y\in(0,L)$ with initial values
\begin{subequations}
\label{eq:ODE-simplified-III-bc}%
\begin{align}
\label{eq:ODE-simplified-III-bc-rho_f-L}
 & \rho_f(L) =  \frac{p_{\text{HG}}}{R\,T_f(L)},\\[2mm]
\label{eq:ODE-simplified-III-bc-v-L}
 & v(L) = \frac{\mc}{\Ac} \frac{R\,T_f(L)}{p_{\text{HG}}}
\end{align}
\end{subequations}
and coefficient
\begin{equation}
\label{eq:ODE-N-III}
 \oN(\rho_f,v,T_f,T_s) :=
 \frac{R (\cpf\mc)^{-1} \hv\Ac \rho_f (T_s-T_f) + K_D^{-1}\mu_f v + K_F^{-1}\rho_f v^2}{\rho_f(\varphi^{-2} v^2 - R T_f)}.
\end{equation}

By means of the mass conservation~\eref{eq:mass-cons-III}, we can substitute the velocity by the density in~\eref{eq:ODE-N-III}. Since we already know the temperatures, 
the coefficient $\oN$ 
becomes a function that only depends on the density and the position, whereas the temperatures enter as parameters:
\begin{equation}
\label{eq:N_quer}
\oN(y,\rho_f;T_f, T_s) =
\frac{R (\cpf\,\mc)^{-1} \hv \Ac \rho_f^2(y)(T_s(y)-T_f(y)) + \mc\Ac^{-1} \left( K_D^{-1} \mu_f + K_F^{-1}\mc\Ac^{-1}\right)}{\varphi^{-2}(\mc \Ac^{-1})^2 - R T_f(y) \rho_f^{2}(y)} .
\end{equation}
We then solve the backward initial value problem
\begin{subequations}
\label{eq:density-problem-iwp-phg-III}%
\begin{align}
\label{eq:density-problem-ode-phg-III}
  &\rho_f'(y) = \oN(y,\rho_f;T_f, T_s)\rho_f(y),\qquad y\in (0,L),\\[2mm]
\label{eq:density-problem-ic-phg-III}
  &  \rho_f(L) =  \frac{p_{\text{HG}}}{R\,T_f(L)}
\end{align}
\end{subequations}
with an ODE solver from top ($y=L$) to bottom ($y=0$). Note that $T_f(L)$ can be computed from \eref{eq:A-simplified-III-solution-temp-Tf} and $p_{\text{HG}}$ is assumed to be given. If this initial value problem has a unique solution, we may determine the reservoir pressure $p_{\text{R}}$ by
\begin{align}
  \label{eq:pr-III}
  p_{\text{R}} = \rho_f(0) \, R \,T_b .
\end{align}

The solution to the initial value problem \eref{eq:density-problem-iwp-phg-III} is not known explicitly. However, it can be verified that for arbitrary $L>0$ there exists a unique solution provided that the heat flux at the interface is positive, i.e., \eref{eq:A-assumption-temp-III} holds, and 
\begin{align}
\label{eq:assumption-r-III}
   \rho_f(L) 
    > 
   \sqrt{ \frac{\varphi^{-2}(\mc\,\Ac^{-1})^2}{R\,T_b} }
 =   \frac{\mc}{\varphi\Ac} \frac{1}{\sqrt{R\,T_b}}.
\end{align}


\begin{proposition}[Existence and uniqueness of the density]~\\[2mm]
\label{prop:A-dens-simple-III}
Let $T_f$, $T_s$ be the unique solution of 
the temperature system
\eref{eq:temp-system-full} -- \eref{eq:bc-complex}
determined by 
\eref{eq:Tfs}.
Let the data $q_{\text{HG}}$ and $p_{\text{HG}}$ be chosen such that \eref{eq:A-assumption-temp-III} and \eref{eq:assumption-r-III} hold.
Then
the  initial value problem \eref{eq:density-problem-iwp-phg-III} has a unique solution in
$[0,L]$ for arbitrary $L>0$.
In particular, the density is positive and strictly monotonically decreasing.
\end{proposition}
\begin{proof}
First of all, we perform the  change of coordinates $\oy(y):=L-y$, $y\in[0,L]$, and the change of variables $\orho(\oy):=\rho(y(\oy))$ to rewrite  \eref{eq:density-problem-iwp-phg-III} as forward problem:
\begin{subequations}
\label{eq:A-density-problem-iwp-phg-forward-III}%
\begin{align}
\label{eq:A-density-problem-ode-phg-forward-III}
  &\orho_f'(\oy) = \ooN(\oy,\orho_f;\oT_f, \oT_s) \, \orho_f(\oy) =:\of(\oy,\orho_f(\oy)) ,\qquad \oy\in (0,L),\\[2mm]
\label{eq:A-density-problem-ic-phg-forward-III}
  &  \orho_f(0) =  \frac{p_{HG}}{R\,T_f(L)} =: \orho_0 ,
\end{align}
\end{subequations}
where $\ooN(\oy,\orho_f;\oT_f, \oT_s) := \oN(y(\oy),\rho_f;T_f, T_s)$ and
$\oT_f(\oy):=T_f(y(\oy))$,  $\oT_s(\oy):=T_s(y(\oy))$. Since $T_f$ and $T_s$ are 
solutions of the initial value problem 
 \eref{eq:temp-system-full} 
satisfying in particular \eref{eq:PM_T_f_1D-simple} and \eref{eq:PM_T_f_1D-full}, it holds
\begin{align}
\label{eq:A-cond-rho-III}
  \oT'_f(\oy) =
  \frac{\hv\,\Ac}{\cpf\,\mc} \left( \oT_f(\oy)-\oT_s(\oy)\right) =
  \frac{\hv\,\Ac}{\cpf\,\mc} \left( T_f(y(\oy))-T_s(y(\oy))\right) =
  - T_f'(y(\oy))
  \le  0.
\end{align}
Here non-negativity is ensured by  
Proposition  \ref{prop:A-temp-full-III} 
because of assumption \eref{eq:A-assumption-temp-III}. We now rewrite $\ooN$ as
\[
   \ooN(\oy,\orho_f;\oT_f, \oT_s)  =
   \frac{\mc\,\Ac^{-1}\left( K_D^{-1}\,\mu_f+ K_F^{-1}\mc \Ac^{-1}\right)- R \oT'_f(\oy)\orho_f^2(\oy)}{R\,\oT_f(\oy) \orho_f^2(\oy) - \varphi^{-2}(\mc\,\Ac^{-1})^2}
   =:\frac{\oZ(\oy)}{\oD(\oy)} .
\]
Note that rewriting the initial value problem as forward problem is essential because data are given at $y=L$ rather than $y=0$.\\
To verify the unique existence of the solution to \eref{eq:A-density-problem-iwp-phg-forward-III},
we proceed now similarly to the proof of the local Picard-Lindel\"of theorem using the contraction mapping theorem. We recall this proof here because for our problem at hand we may exploit some problem-inherent properties, e.g., monotonicity. Repeated application of the local Picard-Lindel\"of theorem will finally  lead to a global result.  We will proceed in three steps.\\
\textbf{Step 1}. First of all, we remind that the space $C([0,L]$ of continuous functions on $[0,L]$ equipped with the supremum norm is a Banach space. For a subspace, we consider the set $B_r([0,L]):=\{\orho\in C([0,L]) \,:\, \orho_0 \le \orho(\oy) \le \orho_0+r,\ \oy\in [0,L] \}$ for some $r>0$.
Since $T_f\in C^1([0,L])$ and is positive, 
see Proposition \ref{prop:A-temp-full-III}, 
there exist bounds $0 = T'_m < T'_M$ such that
$T'_m \le -\oT_f'(\oy) = T_f'(y(\oy)) \le T_M'$ for $\oy\in [0,L]$. Furthermore, the temperature is bounded by
$T_b \le \oT_f(\oy) = T_f(y(\oy)) \le T_f(L)$ due to monotonicity of $T_f$.
Note that $T'_M=0$ and $ T_f(L)=T_b$ if ``$=$'' holds in \eref{eq:A-assumption-temp-III}.
Then we may estimate the numerator $\oZ$
and the denominator  $\oD$ for arbitrary $\orho\in B_r([0,L])$ by
\begin{align}
  & \oZ_m(\orho_0,r):=R\,T_m'\orho_0^2+\beta \le \oZ(\oy) \le R\,T_M' (\orho_0+r)^2 +\beta =: \oZ_M(\orho_0,r) , \nonumber\\
  & \oD_m(\orho_0,r):= R\,T_b \orho_0^2 - \gamma \le \oD(\oy) \le R\, T_f(L) (\orho_0+r)^2 - \gamma =: \oD_M(\orho_0,r)
\end{align}
with $\beta:=\mc\,\Ac^{-1}\left( K_D^{-1} \mu_f + K_F^{-1}\mc \Ac^{-1}\right)>0$
and $\gamma:= \varphi^{-2}(\mc\,\Ac^{-1})^2>0$. The numerator is always positive because $T'_m=0$. By assumption \eref{eq:assumption-r-III}, the lower bound of the denominator is also positive provided that
\begin{align}
  \orho_0 \ge \orho(0) = \frac{p_{\text{HG}}}{R\,T_f(L)} .
\end{align}
Thus,  the right-hand side of the ODE
\eref{eq:A-density-problem-ode-phg-forward-III} is positive and can be estimated by
\[
  0 < \of(\oy,\orho) \le (\orho_0+r) \frac{\oZ_M(\orho_0,r)}{\oD_m(\orho_0,r)} =: M(\orho_0,r),\
  \oy\in [0,L],\
  \orho \in B_r([0,L]) .
\]
Since the denominator is bounded away from zero and $\oT_f\in C^1([0,L])$, the function $\of:[0,L]\times B_r([0,L])\to\R_+$
is continuous in both arguments and differentiable in the second argument. In particular, $\of$ is Lipschitz-continuous in the second argument because by the mean value theorem we have
\[
  \left|\of(\oy,\orho) - \of(\oy,\oorho)\right| \le
  \sup_{s\in[0,1]} \left\{
    \left|\frac{\partial\,\of}{\partial\,\orho}(\oy,\orho+s\,(\oorho-\orho)) \right|
                  \right\}
    |\orho - \oorho| \le K(\orho_0,r) |\orho - \oorho|
\]
for all  $\orho,\overline{\orho}\in [ \orho_0, \orho_0+r]$ and  $\oy\in[0,L]$. Note that $B_r([0,L])$ is a convex set. To determine $K(\orho_0,r)$, we first determine the derivative
\[
\frac{\partial\,\of}{\partial\,\orho}(\oy,\orho) =
\frac{\gamma\left(3\,R\,\oT_f'(\oy)\,\orho^2-\beta \right) -R\,\oT_f(\oy)\,\orho^2\left(R\,\oT_f'(\oy)\,\orho^2+\beta \right)}{\left( R\,\oT_f(\oy)\,\orho^2-\gamma \right)^2} .
\]
This can be estimated by
\[
\left| \frac{\partial\,\of}{\partial\,\orho}(\oy,\orho) \right| \le
\frac{\gamma\left( 3\,R\,T_M'\,(\orho_0+r)^2 + \beta\right) + R\,T_f(L)\,(\orho_0+r)^2 \oZ_M(\orho_0,r)  }{\oD_m(\orho_0,r)^2} =: K(\orho_0,r).
\]
\\
\textbf{Step 2}.
Now we are ready to prove the local existence of a unique solution to the initial value problem~\eref{eq:A-density-problem-iwp-phg-forward-III}  on $[\oy_0,\oy_0+\delta_0] \subset [0,L]$ with $\oy_0:=0$
and $\delta_0:= \min(L,r/M(\orho_0,r),1/(K(\orho_0,r)+\epsilon))$
for some $\varepsilon>0$ arbitrarily small but fixed.
Note that the estimates in Step 1 also hold for the domain $[\oy_0,\oy_0+\delta_0]$ instead of $[0,L]$. We now apply
the contraction mapping theorem to the fixed point problem: find $\orho\in B_r([\oy_0,\oy_0+\delta_0])$ such that $\Phi(\orho) = \orho$, where
the fixed point function $\Phi:B_r([\oy_0,\oy_0+\delta_0]) \to C([\oy_0,\oy_0+\delta_0])$ is defined by
\[
   \Phi(\orho)(\oy) := \orho_0 + \int_{\oy_0}^\oy \of(s,\orho(s))\, ds, \ \oy\in[\oy_0,\oy_0+\delta_0] .
\]
First of all, we note that $\Phi$ maps $B_r([0,\delta_0])$ onto itself. Obviously, the primitive function of a continuous function is continuous. Since by assumption \eref{eq:assumption-r-III} the function $\of(\oy,\orho)$ is positive and bounded by $M(\orho_0,r)$ for any function
$\orho\in B_r([0,\delta_0])$, we may estimate for $\oy\in [\oy_0,\oy_0+\delta_0]$
\[
  \orho_0 \le \Phi(\orho)(\oy)= \orho_0 + \int_{\oy_0}^\oy \of(s,\orho(s))\, ds \le
  \rho_0 + \delta_0 M(\rho_0,r) \le \orho_0 + r .
\]
Next we show that $\Phi$ is a contractive mapping. Fur this purpose, we first note that
\[
  |\Phi(\orho)(\oy) - \Phi(\oorho)(\oy)| \le
   \int_{\oy_0}^\oy |\of(s,\orho(s))-\of(s,\oorho(s))|\, ds \le
  \frac{K(\rho_0,r)}{K(\rho_0,r)+\varepsilon}\Vert \orho-\oorho\Vert_{C([\oy_0,\oy_0+\delta_0])}
\]
for $\oy\in [\oy_0,\oy_0+\delta_0]$ and, thus,
\[
  \Vert \Phi(\orho) - \Phi(\oorho)\Vert_{C([\oy_0,\oy_0+\delta_0])} \le \frac{K(\rho_0,r)}{K(\rho_0,r)+\varepsilon}\Vert \orho-\oorho\Vert_{C([\oy_0,\oy_0+\delta_0])} .
\]
The contraction mapping theorem then ensures the unique existence of $\orho\equiv \orho(\cdot,[\oy_0,\oy_0+\delta_0])\in B_r([\oy_0,\oy_0+\delta_0])$ solving the initial value problem \eref{eq:A-density-problem-iwp-phg-forward-III} on $[\oy_0,\oy_0+\delta_0]$.
The positivity of $\of$  on $[\oy_0,\oy_0+\delta_0]\times B_r([\oy_0,\oy_0+\delta_0])$ implies
that the derivative of $\orho$ is positive and, thus, the density $\orho$ is strictly monotonically increasing. Furthermore, since the initial data are positive, i.e., $\orho_0>0$, the density must be positive as well.\\
\textbf{Step 3}.
In general, we 
will
not find $r>0$ such that $\delta_0= L$. Therefore, we
 repeat  the above local procedure where we replace
$\oy_0$, $\orho_0$, $\delta_0$ by $\oy_{i+1}:=\oy_i+\delta_i$, $\orho_{i+1}:= \orho(\oy_{i+1};[\oy_i,\oy_{i+1}])$ and $\delta_{i+1}:= \min(L-\oy_{i+1},r/M(\orho_{i+1},r),1/(K(\orho_{i+1},r)+\epsilon))$ for $i\in\N_0$. Note that due to monotonicity of the solution $\orho(\cdot;[\oy_i,\oy_{i+1}])$  it holds $\orho_i \le \orho_{i+1} \le \orho_i +r$  and therefore $\rho_0 \le \orho_{i+1}  \le \orho_0 + (i+1)\,r$.
The crucial question is whether we can reach the point $\oy=L$ after a \emph{finite} number of iterations, i.e., does there exist $n\in\N_0$ such that
$\sum_{i=0}^n \delta_i = L$. \\
First of all, we observe that $K(\orho_{i},r)$ is bounded away from zero. In particular, $K(\orho,r)$ tends to the ratio $T_f(L)\,T_M'/T_b^2>0$ for $\orho\to \infty$. Thus, for $i$ sufficiently large $\delta_i<1/K(\orho_{i},r)$. On the other hand, $M(\orho,r)$ tends to infinity for $\orho\to\infty$. Thus, it might happen that $\delta_i=r/M(\orho_i,r)$
for $i\ge i_0$. We now have to verify that $\sum_{i\ge i_0} r/M(\orho_{i},r)$ diverges.
For this purpose, we first note that the ratio $\oZ_M(\orho,r)/\oD_m(\orho,r)$ tends to the ratio $T'_M/T_b$ for $\orho\to\infty$. Thus, this ratio stays bounded for $\orho\in[\orho_0,\infty)$, i.e., there exist constants $0<c_m \le c_M < \infty$ such that
$c_m\le \oZ_M(\orho,r)/\oD_m(\orho,r)< c_M$ for $\orho\in [\orho_0,\infty)$. Then we can estimate
\[
  \sum_{i=i_0}^n \frac{r}{M(\orho_i,r)} =
  \sum_{i=i_0}^n \frac{r}{\orho_i+r} \frac{\oD_m(\orho_i,r)}{\oZ_M(\orho_i,r)} \ge
  \frac{1}{c_M} \sum_{i=0}^n \frac{r}{\orho_i+r}  \ge
  \frac{1}{c_M} \sum_{i=0}^n \frac{r}{\orho_0+ (i+1)\,r} .
\]
The right-hand side tends to infinity for $n\to\infty$ since the harmonic sum is diverging.
Thus,  $\sum_{i>i_0}^n \delta_i\ge~L$ for $n$ sufficiently large.
This proves uniqueness and existence of the solution to problem \eref{eq:A-density-problem-iwp-phg-forward-III} in $[0,L]$.
Since by definition $\rho_f(y) = \orho(\oy(y))\ge \orho_0$ and $\rho_f'(y) = -\orho'(\oy(y))$,
the problem \eref{eq:density-problem-iwp-phg-III} also has a unique solution that is positive and strictly monotonically decreasing.
\end{proof}


Note that the assumption 
\eref{eq:assumption-r-III}
is crucial for the proof to ensure that 
the coefficient $\oN$ determined by~\eqref{eq:N_quer} is negative. In \cite{RomMueller22_ijhmt}, we verify that 
\eref{eq:assumption-r-III}
holds for our envisaged applications.

\begin{remark}[Density with the simplified temperature system]
In case of the simplified temperature system neglecting fluid heat conduction determined by 
\eref{eq:temperature-system-simple} and corresponding boundary conditions \eref{eq:T_f_reservoir_1D}, \eref{eq:T_s_reservoir_1D} at $y=0$ and \eref{eq:heat_balance_int_1D} at $y=L$, 
the mass-momentum system does not change. Nevertheless, since the temperatures \eref{eq:A-simplified-III-solution-temp} enter the mass-momentum system as parameters, they affect its solution. The proof of Proposition \ref{prop:A-dens-simple-III} employs the properties of the temperatures in Proposition \ref{prop:A-temp-full-III} which are the same for the temperature system with fluid heat conduction and the simplified one without. Therefore, Proposition \ref{prop:A-dens-simple-III} holds true also in case of the simplified system as was originally presented in \cite{RomMueller22_ijhmt}, Appendix~B.
\end{remark}

\textbf{Velocity}. Finally, we can determine the velocity by means of the mass conservation
\eref{eq:mass-cons-III}:
\begin{equation}
\label{eq:velo-simplified-III}
   v(y) =  \frac{\mc}{\Ac\,\rho_f(y)},\quad y\in[0,L] .
\end{equation}
In particular, \eqref{eq:velo_int_1D} holds at $y=L$.
By Proposition 
\ref{prop:A-temp-full-III}
and \ref{prop:A-dens-simple-III}, we then can draw the following conclusion.
\begin{proposition}[Properties of the velocity]~\\[2mm]
\label{prop:velo-simple-III}
Let the assumptions of Proposition \ref{prop:A-dens-simple-III}  hold. Then the velocity  is positive and strictly monotonically increasing.
\end{proposition}
\begin{proof}
Since the density is positive and strictly monotonically decreasing, we conclude from
\eref{eq:velo-simplified-III} the assertion.
\end{proof}

\textbf{Pressure}.
By the ideal gas law $p=\rho_f\, R\, T_f$, the density $\rho_f$ and the temperature $T_f$ of the coolant are linked with the pressure. Thus, the pressure is positive. In particular, the pressure is strictly monotonically decreasing for arbitrary $L>0$ 
provided the heat flux at the interface is positive, i.e., \eref{eq:A-assumption-temp-III} holds, and 
the condition \eref{eq:assumption-r-III} is satisfied.
 This immediately follows from the derivative 
\begin{equation}
p'(y) =
R\, \rho_f(y) \frac{T_f'(y) \varphi^{-2} (\mc\,\Ac^{-1})^2 + \mc\,\Ac^{-1}\left( K_D^{-1}\,\mu_f + K_F^{-1}\,\mc\,\Ac^{-1}\right)}{\varphi^{-2} (\mc\,\Ac^{-1})^2-R\,T_f(y)\,\rho_f^2(y)}
\end{equation}
determined by \eref{eq:density-problem-ode-phg-III} and 
\eref{eq:PM_T_f_1D-full} or \eref{eq:PM_T_f_1D-simple}.
Since the derivative of $T_f$ is positive, the numerator is positive. On the other hand, the denominator is negative due to assumption \eref{eq:assumption-r-III} and $\rho_f(y)\ge \rho_f(L)$. Thus, the derivative of the pressure is negative, i.e., the pressure is strictly monotonically decreasing.

\textbf{Summary}. Let the data $T_b$, $p_{\text{HG}}$, $q_{\text{HG}}$ be chosen such that the estimates
\begin{subequations}
\label{eq:assumptions-simple-III}%
\begin{align}
 \label{eq:assumptions-simple-b-III}
&  q_{\text{HG}} \ge 0,\\
 &  \frac{p_{\text{HG}}}{R\,T_f(L)}
     \equiv \rho_f(L)
 > \frac{\mc}{\varphi \Ac} \frac{1}{\sqrt{R\,T_b}}
\end{align}
\end{subequations}
hold. Then, there exists a unique solution to the model consisting of the mass-momentum system \eref{eq:mass-momentum-system} and
either the temperature system with fluid heat conduction \eref{eq:temperature-system-full} or the simplified temperature system \eref{eq:temperature-system-simple} satisfying the boundary conditions \eref{eq:T_f_reservoir_1D}, \eref{eq:T_s_reservoir_1D} at $y=0$ as well as \eref{eq:heat_balance_int_1D}, \eref{eq:Tf_int_1D} (the latter only for the system with fluid heat conduction) at $y=L$, where the  reservoir pressure $p_{\text{R}}$ is determined by
\begin{equation}
p_{\text{R}} = \rho_f(0) \, R \,T_b.
\end{equation}
In particular, the temperatures $T_f$ and $T_s$ are determined by \eref{eq:Tfs} or \eref{eq:A-simplified-III-solution-temp}.  The density $\rho_f$ is given by the solution of \eref{eq:density-problem-iwp-phg-III}. Finally, the velocity $v$ is determined by \eref{eq:velo-simplified-III}.

\section{Numerical Investigations}
\label{sec:numerical_results}
\begin{table}[!b]
  \centering
  \begin{tabular}{llrrl}
    \toprule
    && 2D test case~1 & 2D test case~2 &\\
    \midrule
    \textbf{porous material} &&&&\\
    length (in flow direction~$y$) & $L$ & $0.015$ & $0.01$ & m\\
    porosity & $\varphi$ & $0.111$ & $0.111$ & -\\
    permeability coefficient & $K_D$ & $3.57 \cdot 10^{-13}$ & $3.57 \cdot 10^{-13}$ & m$^2$\\
    Forchheimer coefficient & $K_F$ & $5.17 \cdot 10^{-8}$ & $5.17 \cdot 10^{-8}$ & m\\
    thermal conductivity solid & $\kappa_s$ & $15.19$ & $12.5$ & W/(m$\,$K)\\
    volumetric heat transfer coeff. & $h_v$ & $10^6$ & $10^6$ & W/(m$^3\,$K)\\
    \midrule
    \textbf{cooling gas (air) conditions} &&&&\\
    coolant mass flow rate & $\dot m_c$ & $1.91$ & $6.557$ & g/s\\
    blowing ratio & $F$ & $0.0053$ & $0.0035$ & -\\
    backside solid temperature & $T_b$ & $321.9$ & $304.2$ & K\\
    thermal conductivity & $\kappa_f$ & $\{\mathbf{0.03}, \mathbf{15.19}\}$ & $\{\mathbf{0.03}, \mathbf{12.5}\}$ & W/(m$\,$K)\\
    dynamic viscosity & $\mu_f$ & $2.1 \cdot 10^{-5}$ & $2.1 \cdot 10^{-5}$ & Pa$\,$s\\
    specific heat capacity & $c_{p,f}$ & $1004.5$ & $1004.5$ & J/(kg$\,$K)\\
    \midrule
    \textbf{hot gas (air) conditions} &&&&\\
    inflow Mach number & $Ma_{\text{HG},\infty}$ & $0.3$ & $0.3$ & -\\
    inflow temperature & $T_{\text{HG},\infty}$ & $425$ & $1{,}800$ & K\\
    inflow pressure & $p_{\text{HG},\infty}$ & $95.2$ & $1{,}000$ & kPa\\
    \bottomrule
  \end{tabular}
  \caption{Parameters for the two test cases.}
  \label{tab:test_cases}
\end{table}
In the following, we numerically investigate whether the use of the simplified temperature system in a coupled transpiration cooling simulation is justified. The usual assumption for neglecting fluid heat conduction is $\kappa_f \ll \kappa_s$. However, we also investigate a case for which this constraint is violated, namely $\kappa_f = \kappa_s$. We use two 2D test cases, whose parameters are summarized in Tab.~\ref{tab:test_cases}, simulating cooling gas injection into a subsonic turbulent hot gas channel flow. The channel has a length of $1.32\,$m and a height of $0.06\,$m. The original porous medium has a square cross section with a side length of $0.061\,$m and a resulting area of $A_c=0.003721\,$m$^2$. It is mounted to the lower channel wall at $x=0.58\,$m. Its length in $y$-direction is $L=0.015\,$m for test case~1 and $L=0.01\,$m for test case~2. The first test case originates from~\cite{Dahmen15_ijhmt}, whereas the second one is motivated by the numerical study in~\cite{Munk19}. Both test cases were also used in~\cite{RomMueller22_ijhmt}. The main differences between the two test cases are in the inflow temperature and inflow pressure of the hot gas flow.

For both test cases, five simulations are carried out as listed in Tab.~\ref{tab:sims_per_test_case}. The first four simulations take into account fluid heat conduction with different values for the thermal conductivity~$\kappa_f$ of the fluid and with different boundary conditions regarding the first-order derivative of the fluid temperature~$T_f$ at the interface $y=L$ as discussed in Sect.~\ref{sec:PM-models}. The fifth simulation neglects fluid heat conduction and, hence, uses the simplified temperature system. Simulations 1, 3 and 5 are conducted with our assembled-1D model presented in~\cite{RomMueller22_ijhmt}, i.e., by solving 1D problems in the porous medium, assembling the 1D solutions to obtain a 2D solution and coupling the latter with the solution of a 2D hot gas flow solver~\cite{Bramkamp04}. In case of the vanishing fluid temperature gradient in $y$-direction (simulations 2 and 4), we do not have at hand an analytical solution for the 1D temperature system. Hence, we use our original-2D approach~\cite{Dahmen14_ijnmf,Dahmen15_ijhmt,Koenig19_jtht,Koenig20_book} applying the same hot gas flow solver but a 2D solver for the porous medium flow. The latter was implemented using the finite element library deal.II~\cite{Bangerth07}.
\begin{table}[t]
  \centering
  \begin{tabular}{rllll}
    \toprule
    & model & fluid heat cond. & $\kappa_f$ & boundary condition at $y=L$\\
    \midrule
    1 & assembled-1D & yes & $\kappa_f \ll \kappa_s$ & $T_f^{\prime}(L)=\frac{\hv \Ac}{\cpf \mc} \left(T_{s}(L)-T_{f}(L)\right) \neq 0$\\
    2 & original-2D  & yes & $\kappa_f \ll \kappa_s$ & $\nabla T_f(x,L) \cdot \vec{n}=0$\\
    3 & assembled-1D & yes & $\kappa_f = \kappa_s$   & $T_f^{\prime}(L)=\frac{\hv \Ac}{\cpf \mc} \left(T_{s}(L)-T_{f}(L)\right) \neq 0$\\
    4 & original-2D  & yes & $\kappa_f = \kappa_s$   & $\nabla T_f(x,L) \cdot \vec{n}=0$\\
    5 & assembled-1D & no  & - & -\\
    \bottomrule
  \end{tabular}
  \caption{Simulations performed for each test case.}
  \label{tab:sims_per_test_case}
\end{table}

Details on the coupled solution procedure, the setup of the solvers, the computational meshes and mesh convergence can be found in~\cite{RomMueller22_ijhmt}. For the adiabatic side walls of the porous medium, in 1D we apply the side wall modeling as proposed and discussed also in~\cite{RomMueller22_ijhmt}. For both test cases, we perform six coupling steps: porous medium flow and hot gas flow solutions are computed six times each in an alternating fashion. After each run of a flow solver, the coupling conditions on the interface are updated. For the chosen test cases, six steps are sufficient for the flow values on the interface to be converged in both domains.

\subsection{Test case 1}
The fluid temperature distribution inside the porous medium is depicted in Fig.~\ref{fig:tc1_Tf_all} for all five simulations of test case~1. The differences between the solutions are small. The largest visible deviation concerns simulation~4 for $\kappa_f=\kappa_s$ with vanishing fluid temperature gradient at $y=L$, see Fig.~\ref{fig:tc1_Tf_sim4}, around the upper left corner. The largest difference related to all other simulations is measured directly at the corner point and refers to simulation~2. However, it is only $5.5\,$K which corresponds to $1.6\,\%$.
\begin{figure}[p]
  \centering
  \begin{subfigure}[t]{0.495\linewidth}
    \centering
    \includegraphics[width=\linewidth]{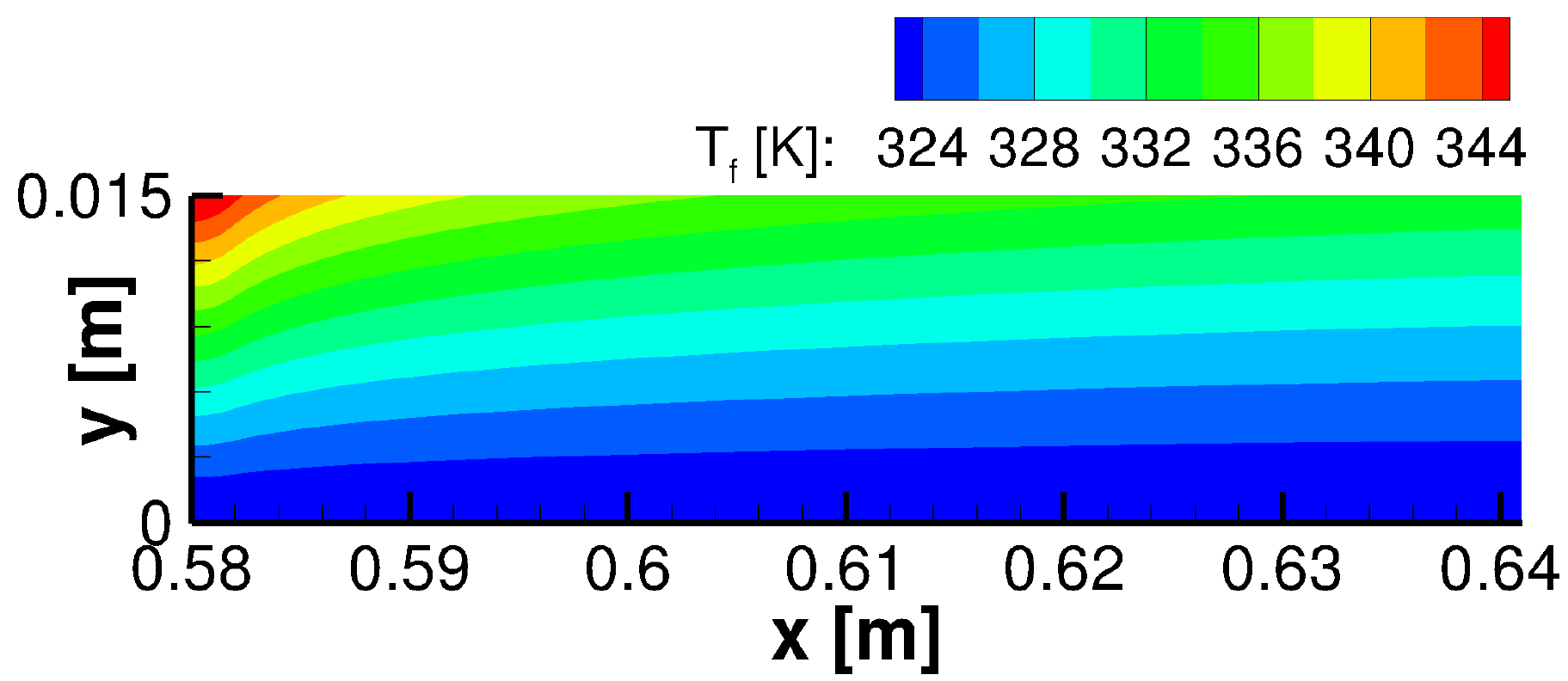}
    \caption{sim.~1: assembled-1D, with fluid heat cond., $\kappa_f \ll \kappa_s$, $T_f^{\prime}(L)=\frac{\hv \Ac}{\cpf \mc} \left(T_{s}(L)-T_{f}(L)\right) \neq 0$}
    \label{fig:tc1_Tf_sim1}
  \end{subfigure}
  \hfill
  \begin{subfigure}[t]{0.495\linewidth}
    \centering
    \includegraphics[width=\linewidth]{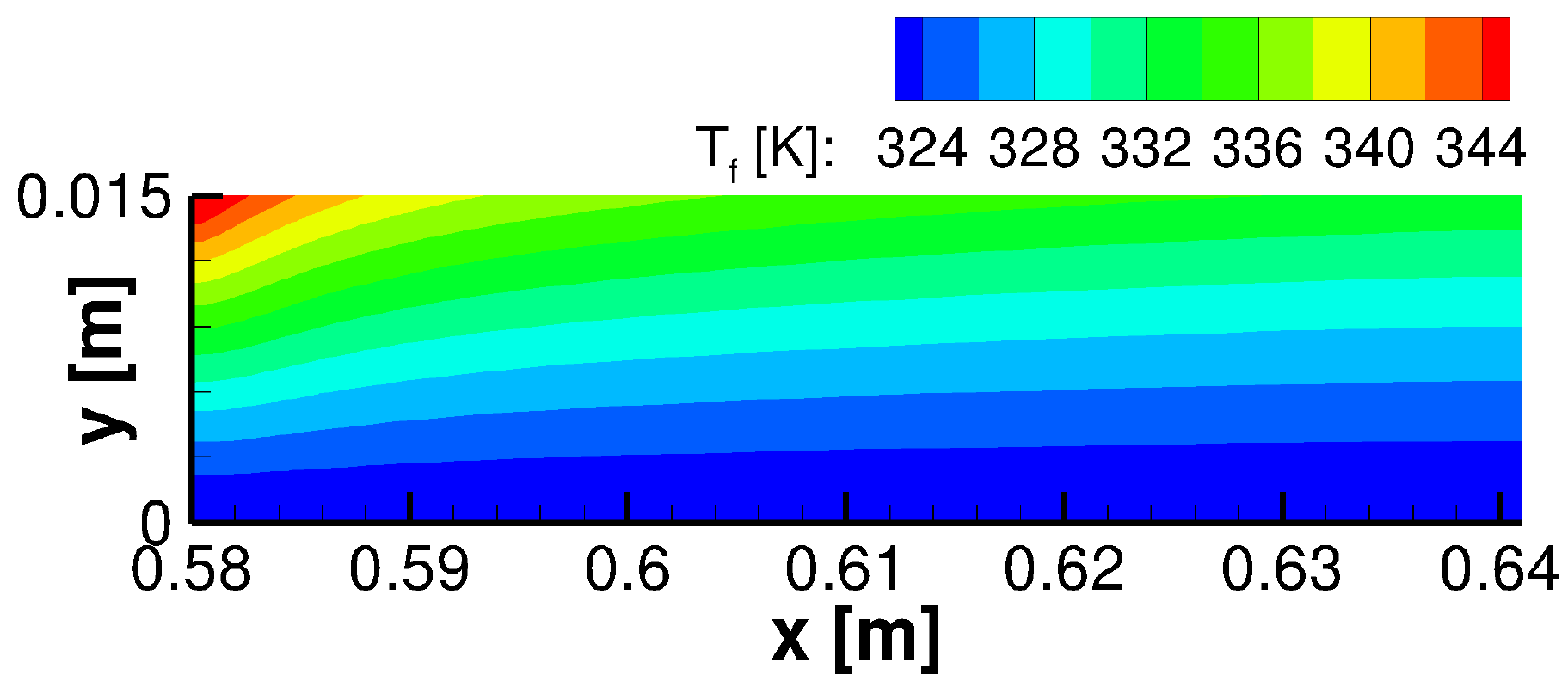}
    \caption{sim.~2: original-2D, with fluid heat cond., $\kappa_f \ll \kappa_s$, $\nabla T_f(x,L) \cdot \vec{n}=0$}
    \label{fig:tc1_Tf_sim2}
  \end{subfigure}
  \\[0.1cm]
  \begin{subfigure}[t]{0.48\linewidth}
    \centering
    \includegraphics[width=\linewidth]{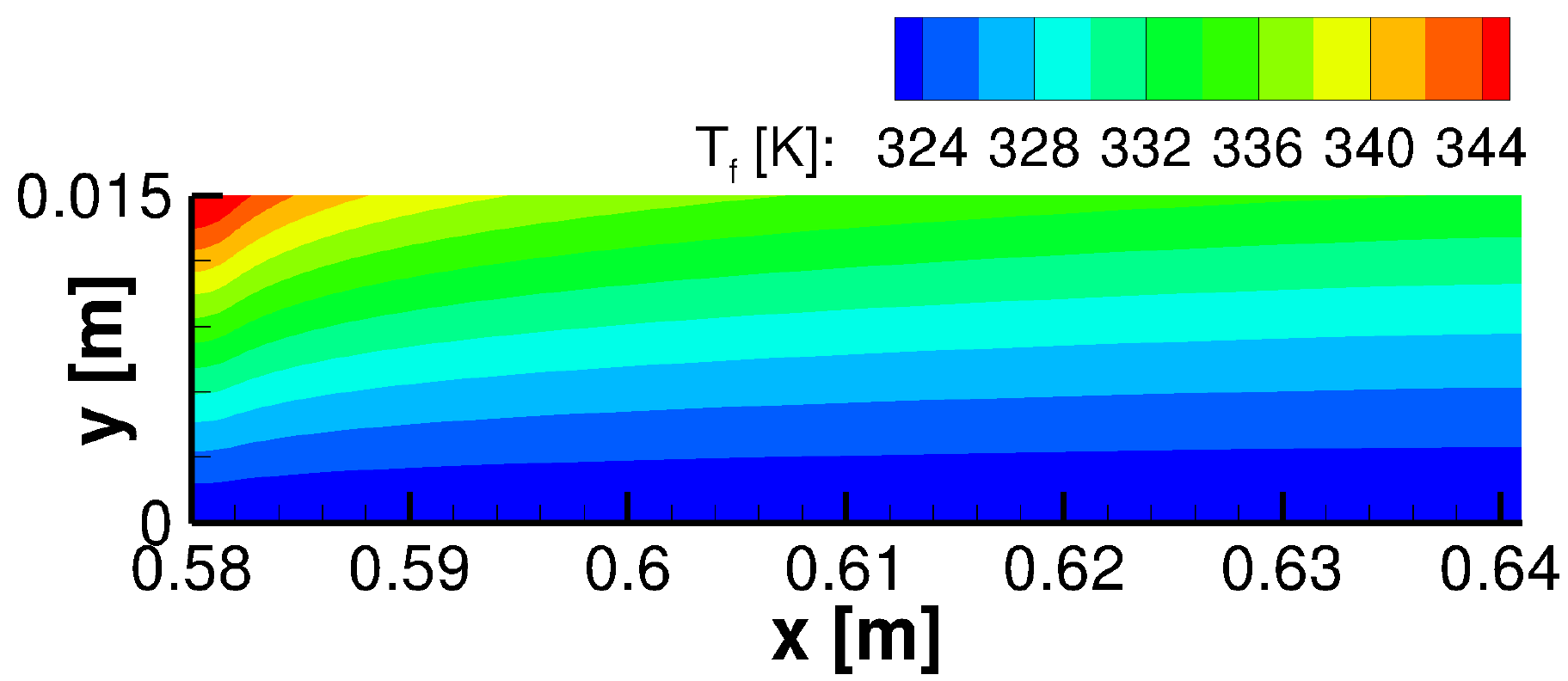}
    \caption{sim.~3: assembled-1D, with fluid heat cond., $\kappa_f = \kappa_s$, $T_f^{\prime}(L)=\frac{\hv \Ac}{\cpf \mc} \left(T_{s}(L)-T_{f}(L)\right) \neq 0$}
    \label{fig:tc1_Tf_sim3}
  \end{subfigure}
  \hfill
  \begin{subfigure}[t]{0.48\linewidth}
    \centering
    \includegraphics[width=\linewidth]{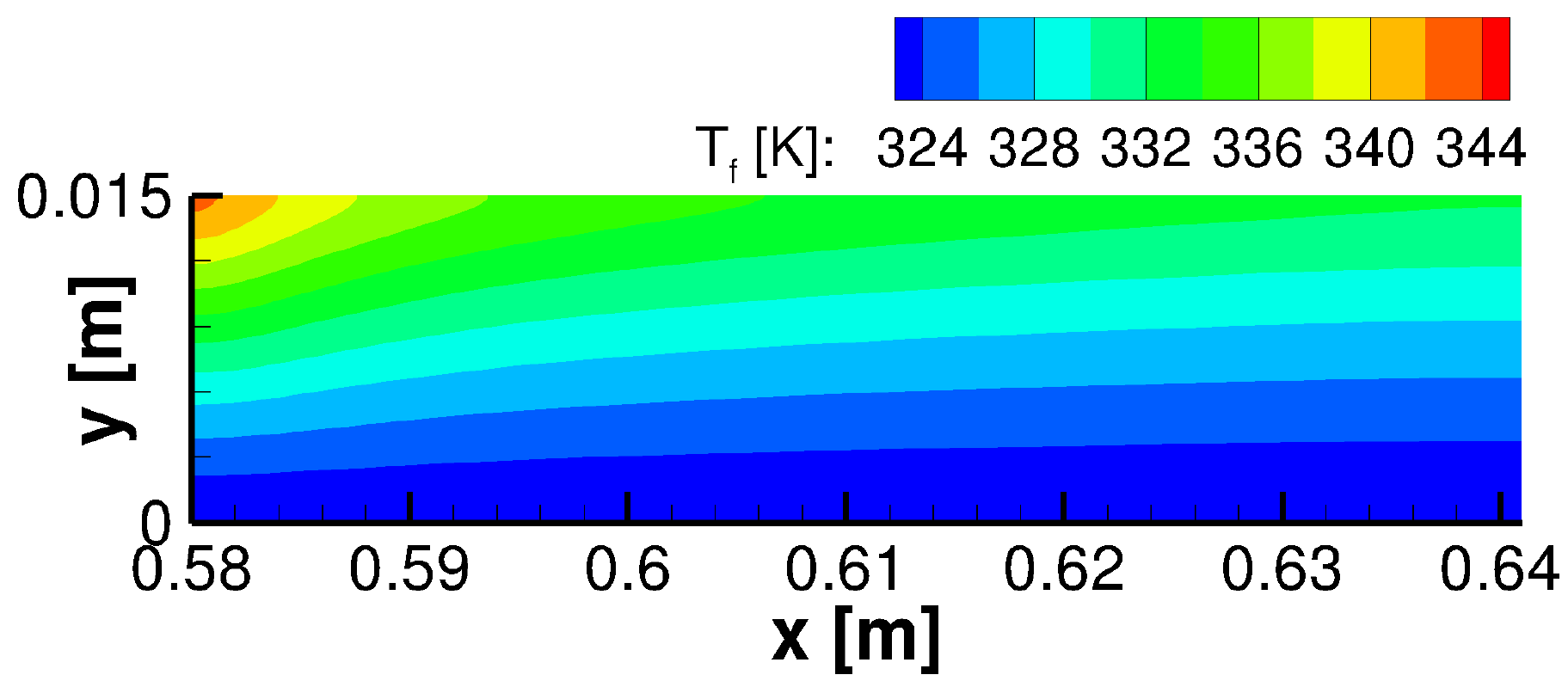}
    \caption{sim.~4: original-2D, with fluid heat cond., $\kappa_f = \kappa_s$, $\nabla T_f(x,L) \cdot \vec{n}=0$}
    \label{fig:tc1_Tf_sim4}
  \end{subfigure}
  \\[0.1cm]
  \begin{subfigure}[t]{0.48\linewidth}
    \centering
    \includegraphics[width=\linewidth]{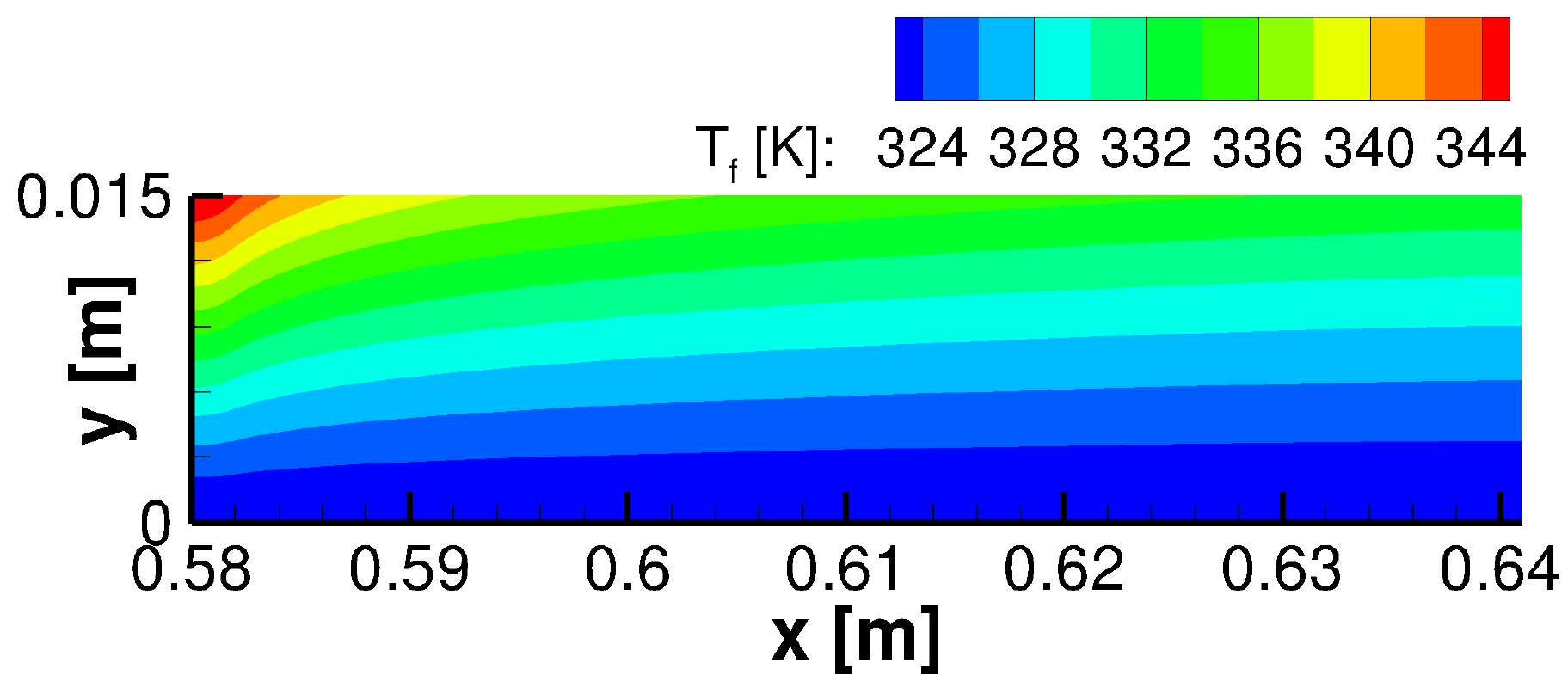}
    \caption{sim.~5: assembled-1D, without fluid heat cond.}
    \label{fig:tc1_Tf_sim5}
  \end{subfigure}
  \caption{Test case 1: Fluid temperature $T_f$ in the porous medium.}
  \label{fig:tc1_Tf_all}
\end{figure}

The influence of the different fluid temperature solutions on the hot gas temperature on the interface, i.e., at the cooling gas injection, and further downstream at the lower channel wall is investigated in Fig.~\ref{fig:tc1_THG_int}. The temperature curves for simulations~1, 2 and 5 are virtually the same. Simulations~3 and 4 for $\kappa_f=\kappa_s$ show minimal differences but only between $x=0.58\,$m and $x=0.641\,m$, i.e., on the interface: the temperatures of simulation~3 are slightly higher and the temperatures of simulation~4 slightly lower than the ones of the three other cases. However, the differences are clearly negligible.
\begin{figure}[p]
  \centering
  \includegraphics[width=0.55\linewidth]{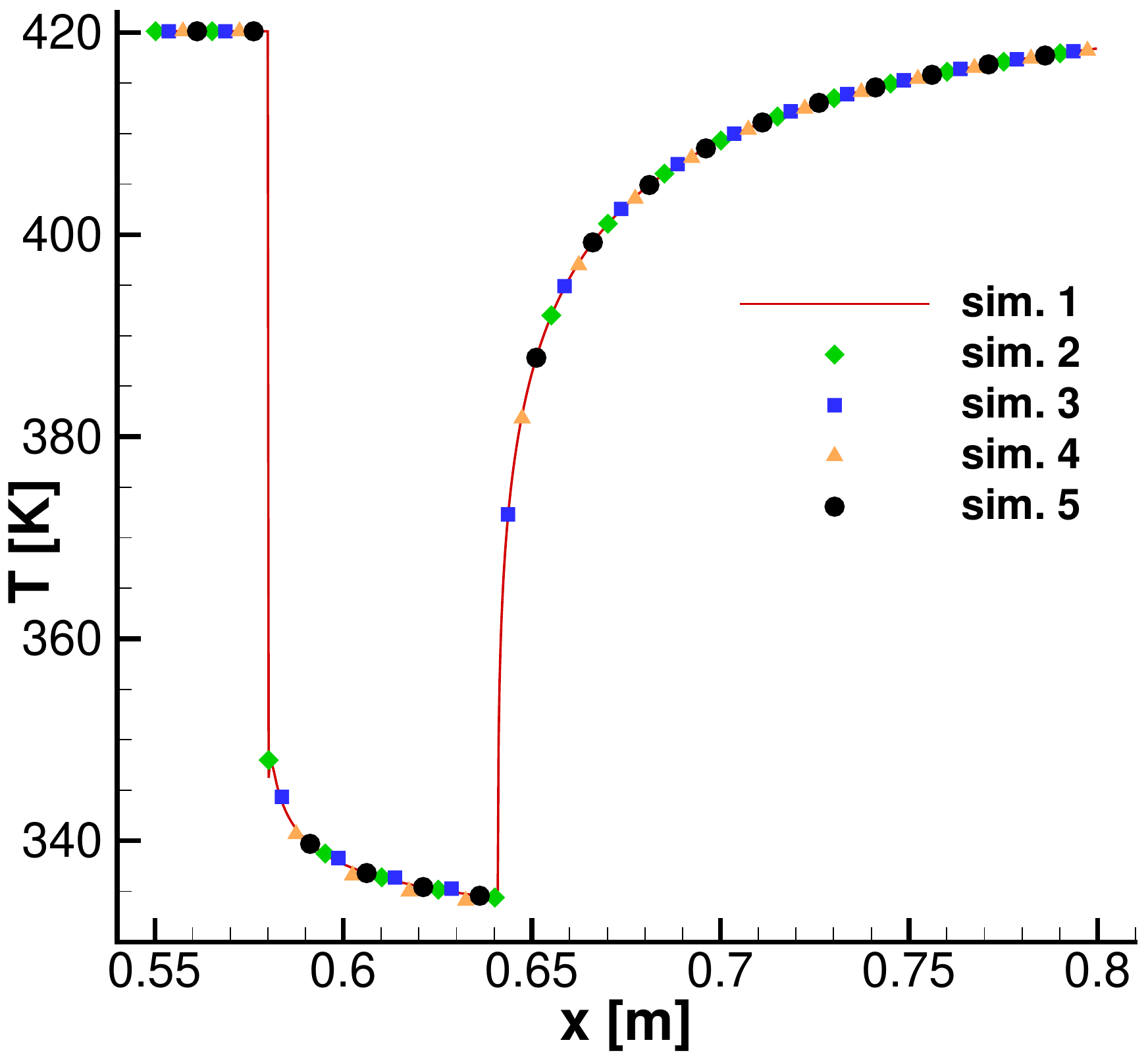}
  \caption{Test case 1: Lower wall hot gas temperature (cooling gas injection at $x \in [0.58, 0.641]\,$m).}
  \label{fig:tc1_THG_int}
\end{figure}

A similar behavior becomes apparent in Fig.~\ref{fig:tc1_PM_x061} showing the temperature throughout the porous medium for $y \in [0,L]$ at fixed position $x=0.61\,$m. While the fluid and solid temperature curves for simulations~1, 2 and 5 coincide, the ones for simulations~3 and 4 with large $\kappa_f$ deviate. Again, the temperatures of simulation~3 are larger than those of simulations~1, 2 and 5, whereas simulation~4 leads to lower temperatures. The adiabatic boundary condition at the interface~$y=L$, which is prescribed in simulations~2 and 4, is only reflected by simulation~4. As explained in Sect.~\ref{sec:PM-models}, this is due to the strong coupling of the two temperature equations and fluid heat conduction being negligible in case of simulation~2. The overall temperature behavior in the porous medium reveals that the choice of the boundary condition for $T_f^{\prime}(L)$ in 1D or $\nabla T_f(x,L) \cdot \vec{n}$ in 2D/3D when taking into account fluid heat conduction has more influence than the choice whether to consider the latter at all. However, with maximum temperature differences of $2.1\,$K ($0.6\,\%$) and $1.2\,$K ($0.4\,\%$) for $T_f$ and $T_s$, respectively, between simulations~3 and 4 at $y=L$, the influence on the cooling of the hot gas flow should be negligible. This is substantiated by Fig.~\ref{fig:tc1_HG_x061}, depicting the hot gas temperature in the boundary layer at the same fixed position $x=0.61\,$m. All five curves are virtually the same with only minor differences in the initial temperature at $y=0\,$m due to the differences in the porous medium at $y=L$.
\begin{figure}[t]
  \centering
  \begin{subfigure}[t]{0.48\linewidth}
    \centering
    \includegraphics[width=\linewidth]{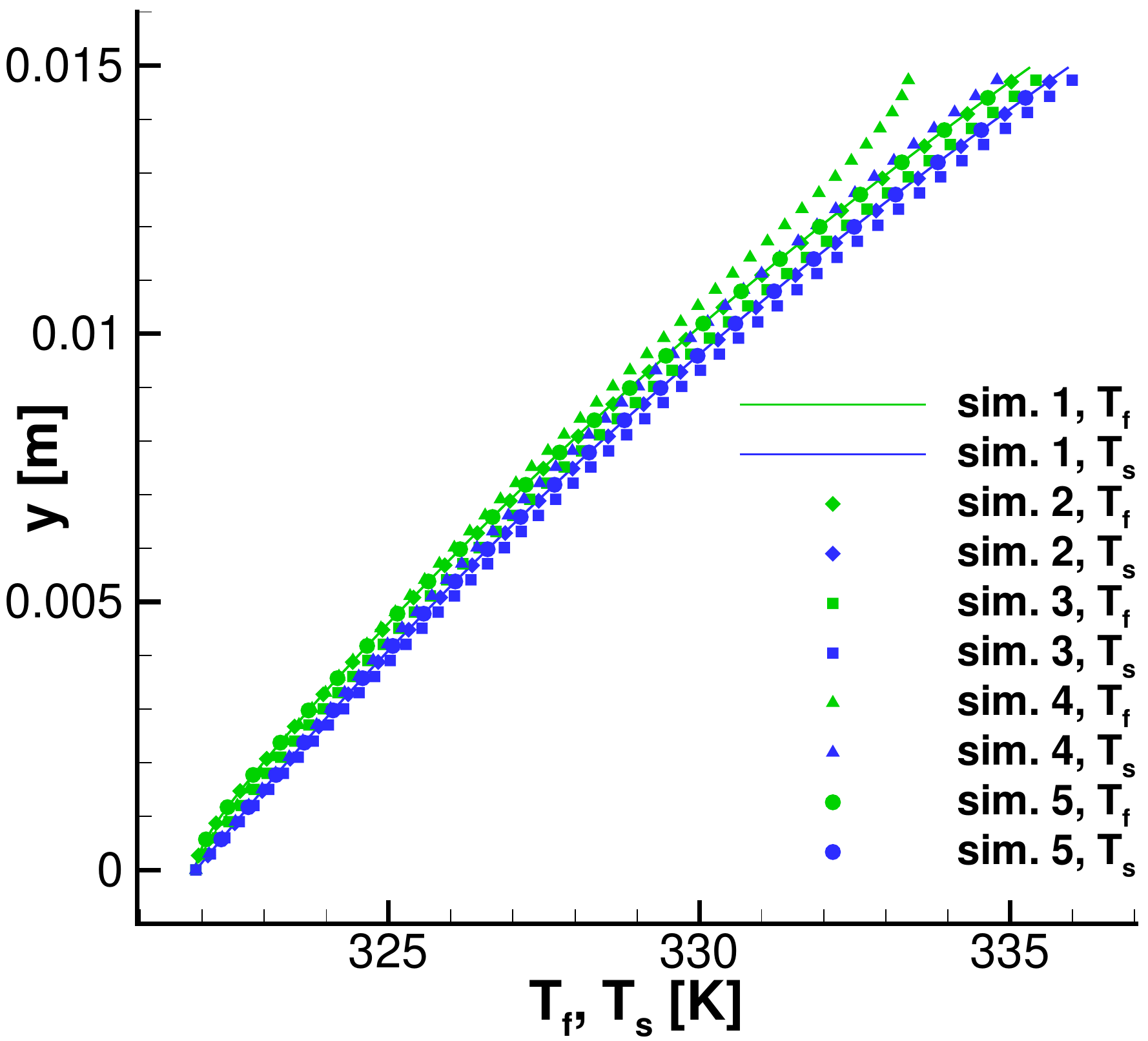}
    \caption{porous medium}
    \label{fig:tc1_PM_x061}
  \end{subfigure}
  \hfill
  \begin{subfigure}[t]{0.48\linewidth}
    \centering
    \includegraphics[width=\linewidth]{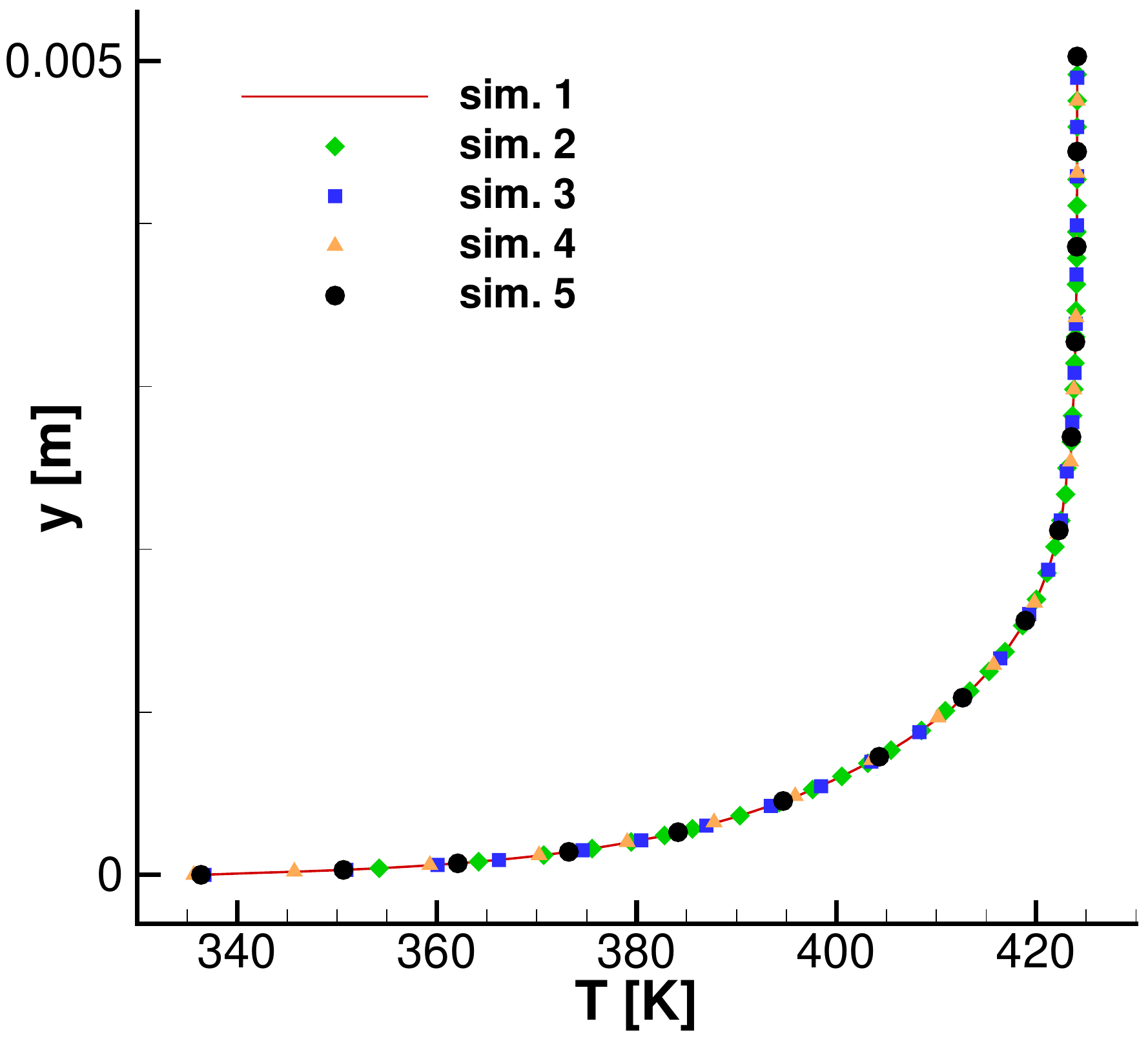}
    \caption{hot gas}
    \label{fig:tc1_HG_x061}
  \end{subfigure}
  \caption{Test case 1: Temperatures at $x=0.61\,$m: (a)~throughout the whole porous medium and (b)~in the boundary layer of the hot gas flow.}
  \label{fig:tc1_x061}
\end{figure}

Since the differences in the fluid density~$\rho_f$ and the Darcy velocity~$v$ regarding the five simulations are analogously small and negligible, we do not present results for these two quantities.

\subsection{Test case 2}
Figures~\ref{fig:tc2_Tf_all} to \ref{fig:tc2_x061} show the results for the high-temperature test case~2 correspondingly to Figs.~\ref{fig:tc1_Tf_all} to \ref{fig:tc1_x061} of test case~1. The main results are analogous to those of test case~1. Again, the largest difference in the fluid temperatures in the porous medium concerns the upper left corner of simulations~2 and 4, cf.~Fig.~\ref{fig:tc2_Tf_all}. At the corner point, the difference is $122.4\,$K which corresponds to $14.7\,\%$. Even though this deviation is much larger than observed for test case~1 ($1.6\,\%$, see above), the influence on the hot gas flow again is small, see Fig.~\ref{fig:tc2_THG_int}. This is due to the differences in the solutions being confined to a very small area around the upper left corner.
\begin{figure}[p]
  \centering
  \begin{subfigure}[t]{0.48\linewidth}
    \centering
    \includegraphics[width=\linewidth]{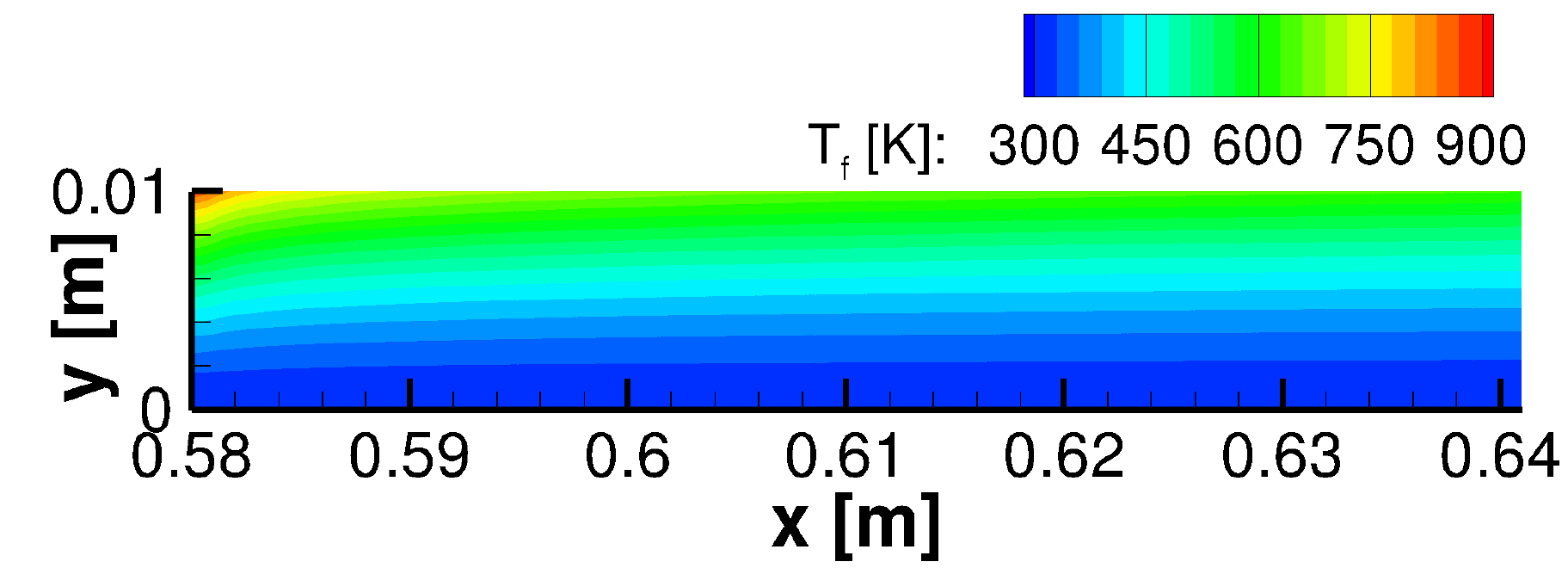}
    \caption{sim.~1: assembled-1D, with fluid heat cond., $\kappa_f \ll \kappa_s$, $T_f^{\prime}(L)=\frac{\hv \Ac}{\cpf \mc} \left(T_{s}(L)-T_{f}(L)\right) \neq 0$}
    \label{fig:tc2_Tf_sim1}
  \end{subfigure}
  \hfill
  \begin{subfigure}[t]{0.48\linewidth}
    \centering
    \includegraphics[width=\linewidth]{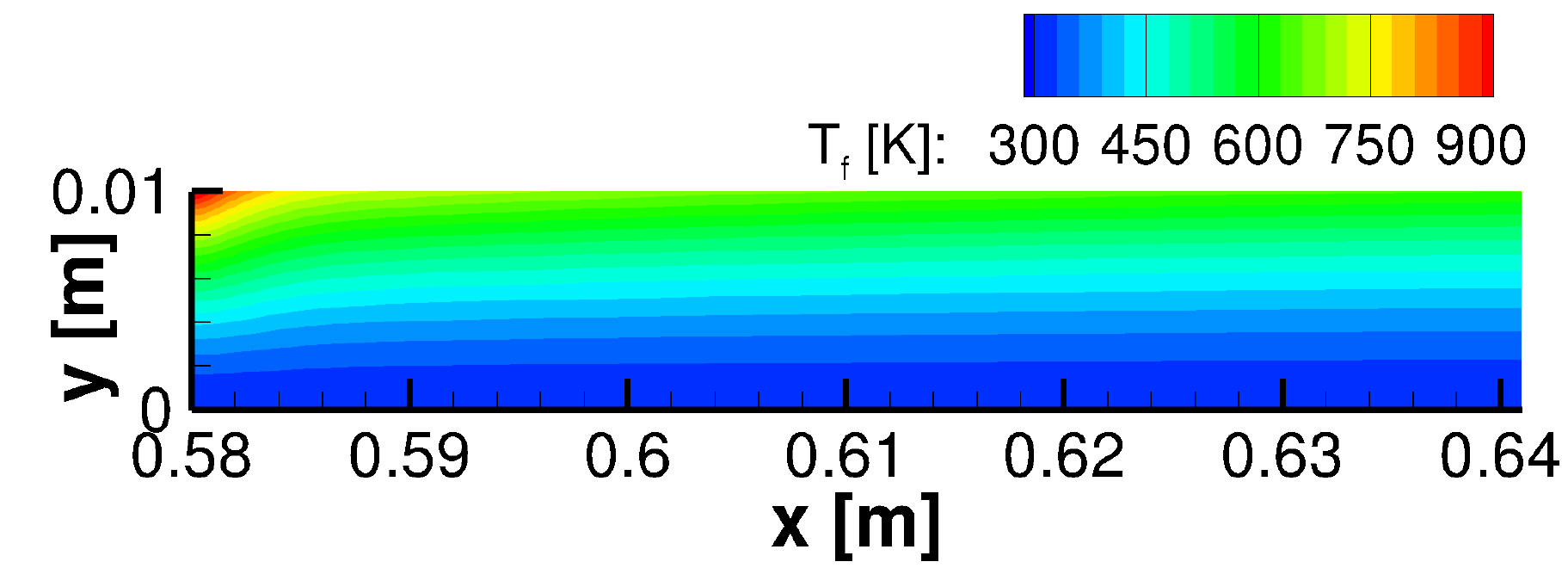}
    \caption{sim.~2: original-2D, with fluid heat cond., $\kappa_f \ll \kappa_s$, $\nabla T_f(x,L) \cdot \vec{n}=0$}
    \label{fig:tc2_Tf_sim2}
  \end{subfigure}
  \\[0.1cm]
  \begin{subfigure}[t]{0.48\linewidth}
    \centering
    \includegraphics[width=\linewidth]{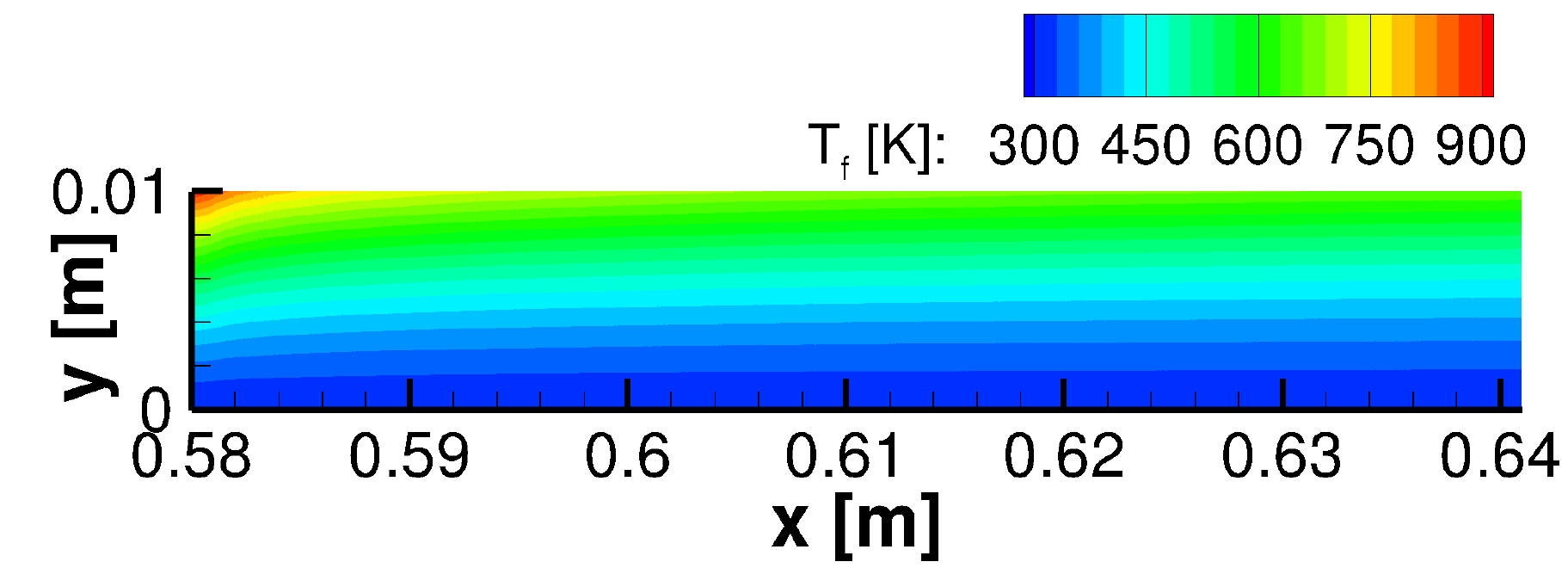}
    \caption{sim.~3: assembled-1D, with fluid heat cond., $\kappa_f = \kappa_s$, $T_f^{\prime}(L)=\frac{\hv \Ac}{\cpf \mc} \left(T_{s}(L)-T_{f}(L)\right) \neq 0$}
    \label{fig:tc2_Tf_sim3}
  \end{subfigure}
  \hfill
  \begin{subfigure}[t]{0.48\linewidth}
    \centering
    \includegraphics[width=\linewidth]{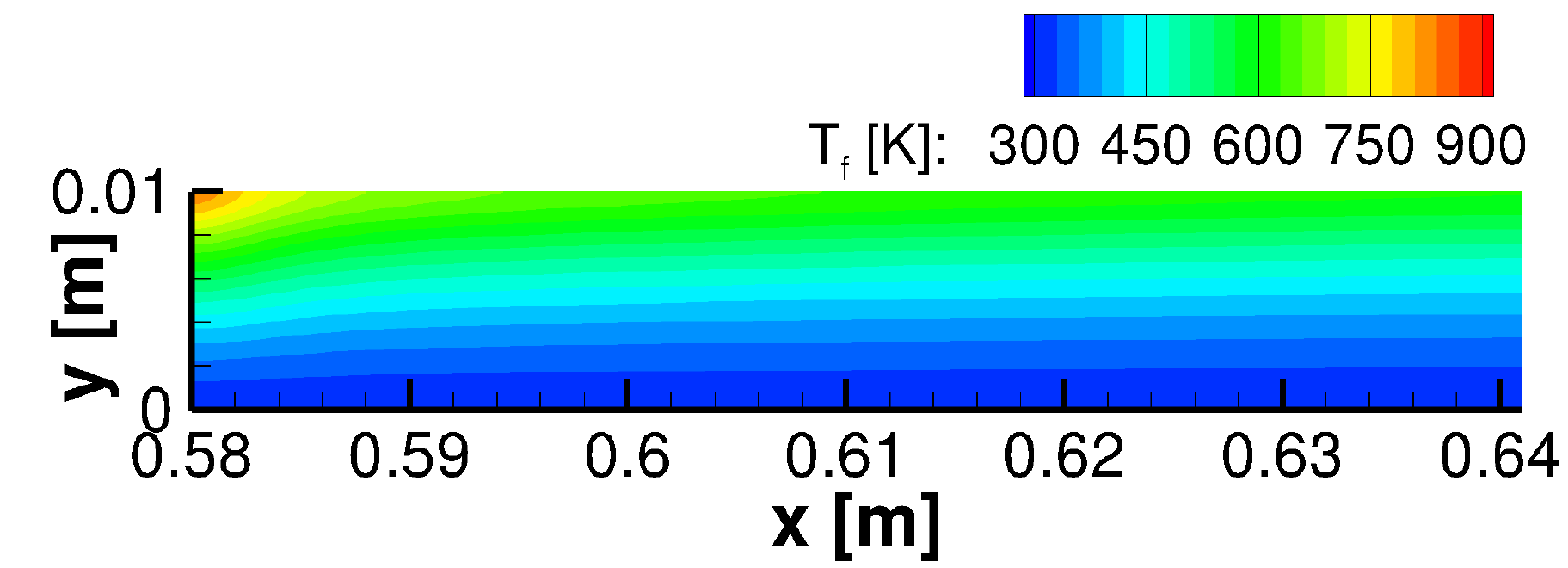}
    \caption{sim.~4: original-2D, with fluid heat cond., $\kappa_f = \kappa_s$, $\nabla T_f(x,L) \cdot \vec{n}=0$}
    \label{fig:tc2_Tf_sim4}
  \end{subfigure}
  \\[0.1cm]
  \begin{subfigure}[t]{0.48\linewidth}
    \centering
    \includegraphics[width=\linewidth]{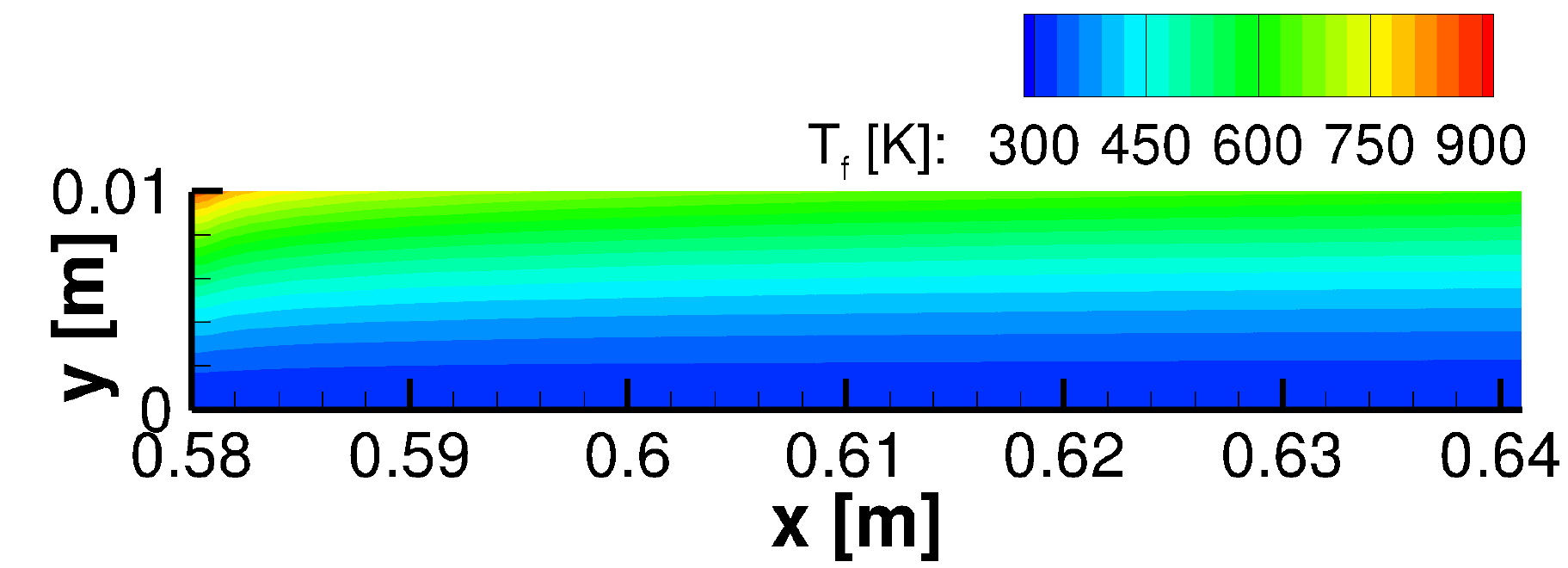}
    \caption{sim.~5: assembled-1D, without fluid heat cond.}
    \label{fig:tc2_Tf_sim5}
  \end{subfigure}
  \caption{Test case 2: Fluid temperature $T_f$ in the porous medium.}
  \label{fig:tc2_Tf_all}
\end{figure}

Figure~\ref{fig:tc2_THG_int} shows that, compared with test case~1, the larger differences in the porous medium temperatures lead to larger differences in the hot gas temperature at the injection. Again, the hot gas temperatures of simulation~3 are the highest and those of simulation~4 the lowest. However, the cooling effect on the lower channel wall downstream of the injection remains virtually the same for all five simulations.
\begin{figure}[p]
  \centering
  \includegraphics[width=0.55\linewidth]{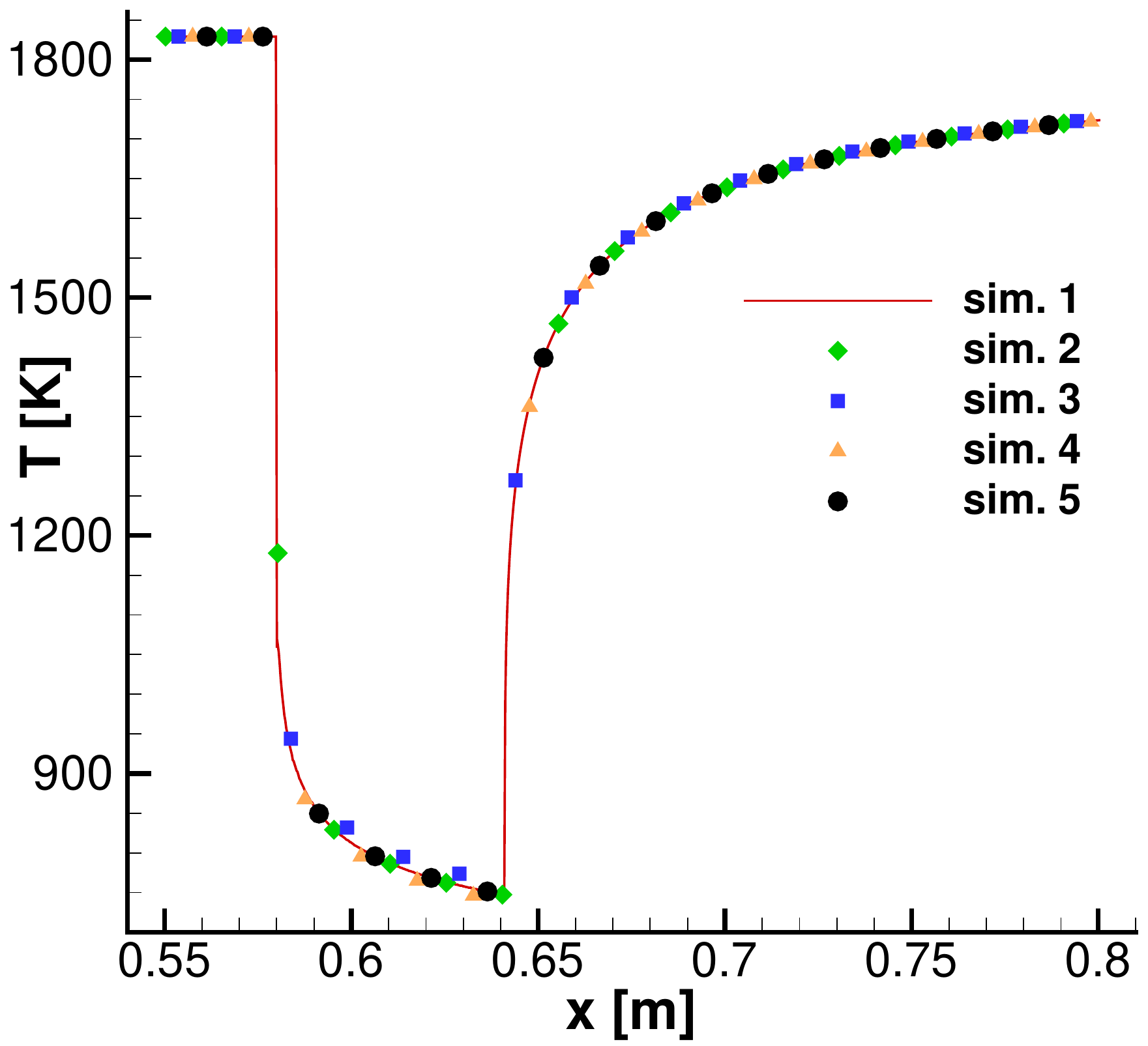}
  \caption{Test case 2: Lower wall hot gas temperature (cooling gas injection at $x \in [0.58, 0.641]\,$m).}
  \label{fig:tc2_THG_int}
\end{figure}

The temperature curves in Fig.~\ref{fig:tc2_PM_x061} first of all reveal a more pronounced temperature nonequilibrium than observed for test case~1, cf.~Fig.~\ref{fig:tc1_PM_x061}. Again, the differences in the temperature solutions are largest at the interface $y=L$. They amount to $45.8\,$K ($7.3\,\%$) and $24.1\,$K ($3.2\,\%$) for $T_f$ and $T_s$, respectively, between simulations~3 and 4. Even though these deviations are significantly larger than those of test case~1 ($0.6\,\%$ and $0.4\,\%$, see above), the boundary layer temperature in the hot gas flow is virtually the same for all five simulations again as Fig.~\ref{fig:tc2_HG_x061} reveals.
\begin{figure}[t]
  \centering
  \begin{subfigure}[t]{0.48\linewidth}
    \centering
    \includegraphics[width=\linewidth]{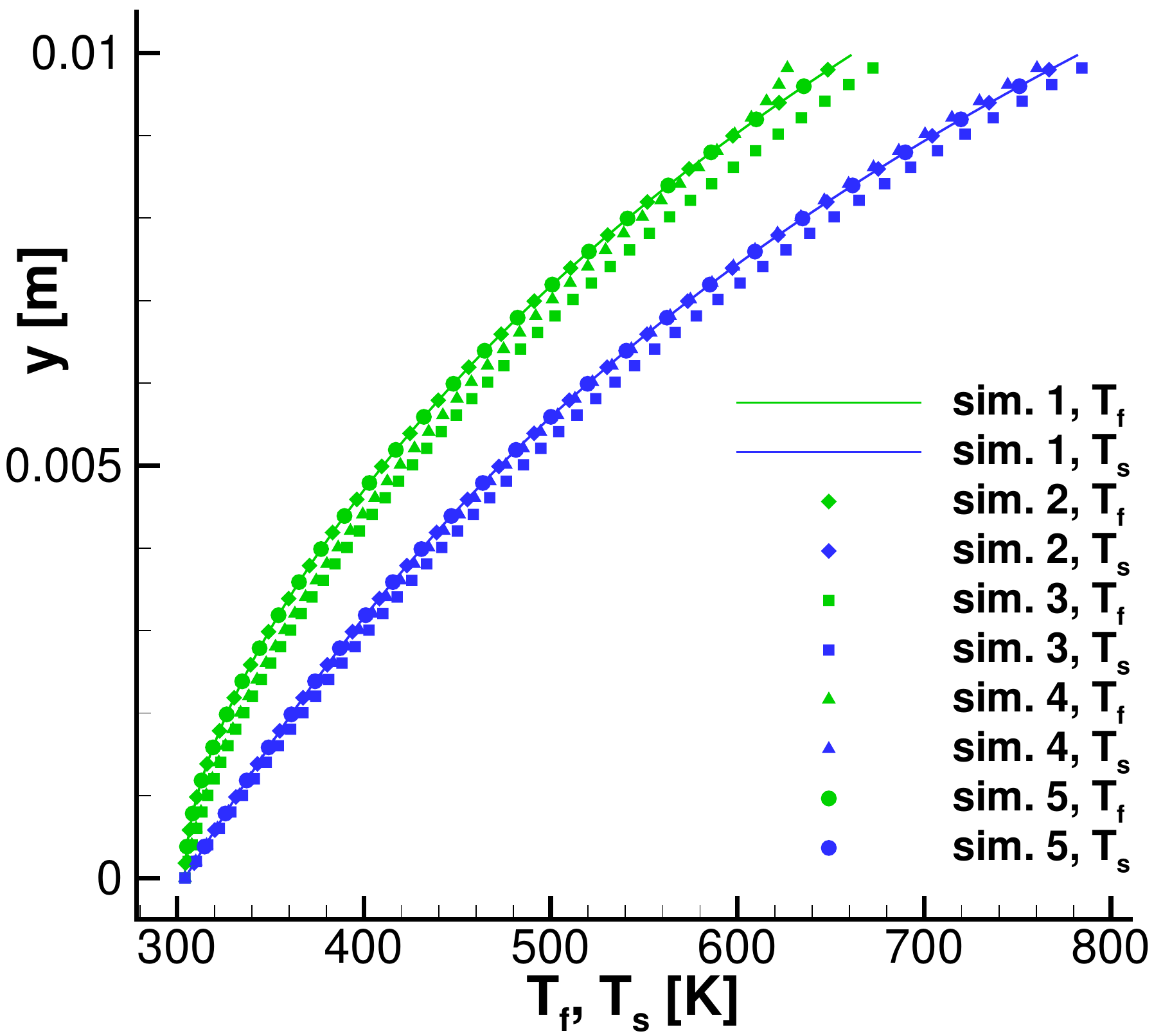}
    \caption{porous medium}
    \label{fig:tc2_PM_x061}
  \end{subfigure}
  \hfill
  \begin{subfigure}[t]{0.48\linewidth}
    \centering
    \includegraphics[width=\linewidth]{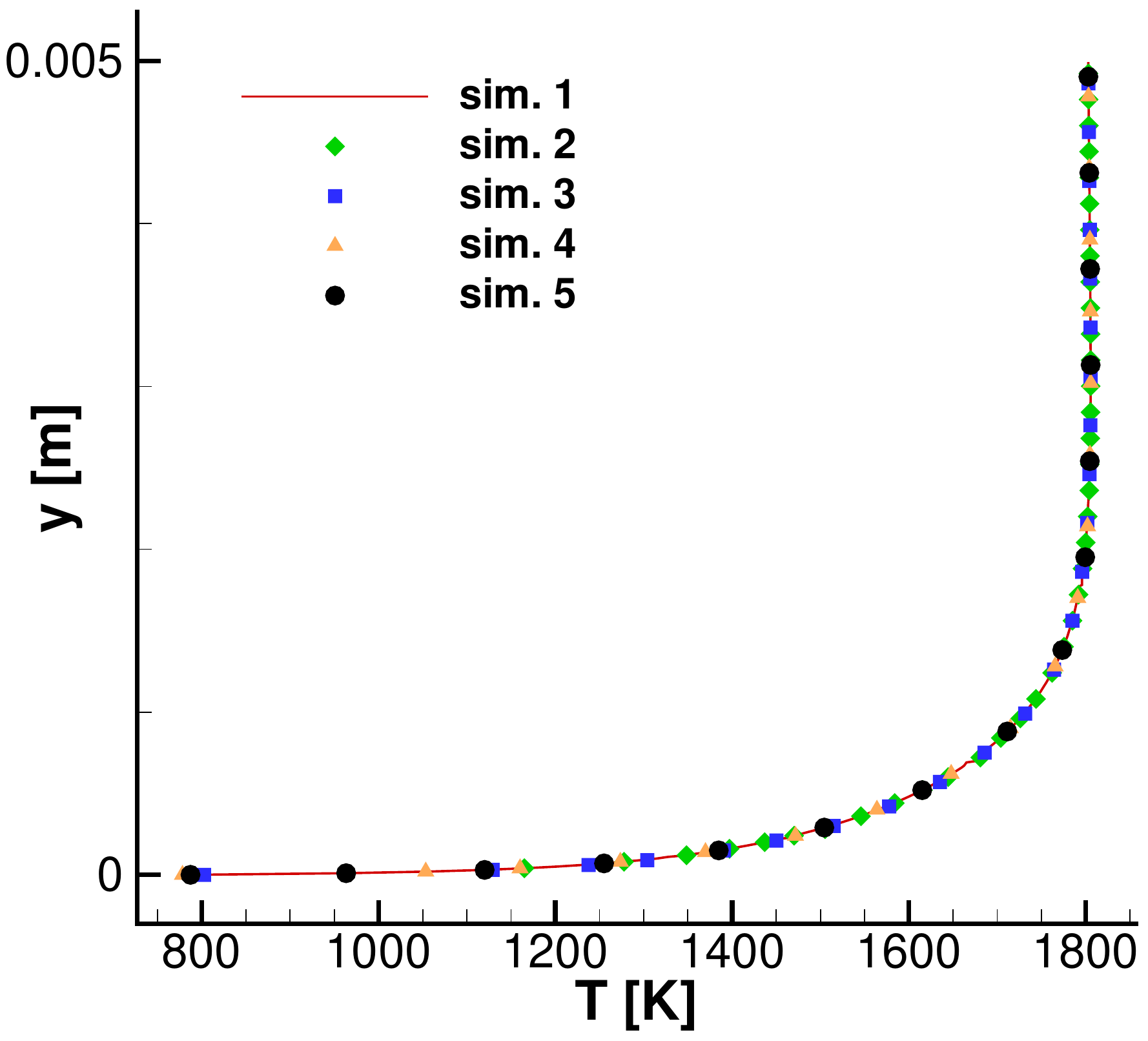}
    \caption{hot gas}
    \label{fig:tc2_HG_x061}
  \end{subfigure}
  \caption{Test case 2: Temperatures at $x=0.61\,$m: (a)~throughout the whole porous medium and (b)~in the boundary layer of the hot gas flow.}
  \label{fig:tc2_x061}
\end{figure}

As the differences in the fluid density~$\rho_f$ and the Darcy velocity~$v$ are also small, we again omit presenting results for these two quantities.


\section{Conclusion}
A one-dimensional porous medium model for Darcy-Forchheimer flow under local thermal nonequilibrium has been investigated taking into account fluid heat conduction. This model consists of a linear temperature system for the fluid temperature and the solid temperature and a nonlinear mass-momentum system for the density of the fluid and the Darcy velocity. The temperature system is solved explicitly. Monotonicity properties of the temperature solution are employed to prove the existence of a unique solution to the mass-momentum system.

The 1D model is embedded in our assembled-1D model in~\cite{RomMueller22_ijhmt} for the simulation of transpiration cooling problems, where a porous medium is mounted into a wall of a hot gas channel. Comparisons are performed with (i) solutions determined using a simplified 1D model omitting fluid heat conduction in the assembled-1D model and (ii) solutions using a full 2D porous medium model with different boundary conditions on the interface regarding fluid heat conduction. Computations are performed for two test cases with moderate and high temperatures in the hot gas channel. The numerical results show a very good agreement of the wall temperature in the hot gas both on top of the porous medium and downstream of the injection surface although there are some differences observable in the porous medium flow solutions near the injection surface. This also holds true if the fluid heat conduction is of the order of the  solid heat conduction.
This justifies the use of the simplified 1D porous medium model in our assembled-1D model for transpiration cooling.


\bibliographystyle{spmpsci}      



\end{document}